\documentclass[journal]{IEEEtran}

\usepackage{amsfonts}
\usepackage{amssymb}
\usepackage{stfloats}
\usepackage{cite}
\usepackage{graphicx}
\usepackage{psfrag}
\usepackage{subfigure}
\usepackage{amsmath}
\usepackage{array}
\usepackage{epstopdf}
\usepackage{authblk}
\usepackage{graphicx} 
\usepackage{amsthm} 
\usepackage{lipsum}
\usepackage{verbatim} 
\usepackage{authblk}
\usepackage{algorithm}
\usepackage{algpseudocode}
\usepackage{amsmath}
\usepackage{graphics}
\usepackage{epsfig}
\usepackage{setspace}

\usepackage[usenames]{color}

\newtheorem{Definition}{Definition}
\newtheorem{Problem}{Problem}
\newtheorem{Lemma}{Lemma}

\newtheorem{Algorithm}{Algorithm}

\newtheorem{Remark}{Remark}

\newtheorem{Baseline}{Baseline}
\newtheorem{Example}{Example}

\newtheorem{Corollary}{Corollary}

\title{Joint Downlink Scheduling for File Placement and Delivery in Cache-Assisted Wireless Networks with Finite File Lifetime}

\author{Bojie Lv, Lexiang Huang, \IEEEmembership{Student Member,~IEEE}, and Rui Wang, \IEEEmembership{Member,~IEEE} 
	\thanks{
		Manuscript received June 24, 2018; revised December 3, 2018 and February 16, 2019; accepted February 18, 2019. This work was supported in part by National Natural Science Foundation of China under grant 61771232, Natural Science Foundation of Guangdong Province of China under grant 2017A030313335 and the Shenzhen Science and Technology Innovation Committee under Grant JCYJ20160331115457945. The associate editor coordinating the review of this paper and approving it for publication was Lawrence Ong. (Corresponding author: Rui Wang.)
		
		Bojie Lv and and Rui Wang are with Department of Electrical and Electronic Engineering, The Southern University of Scienece and Technology, China, and also with the PCL Research Center of Networks and Communications, Peng Cheng Laboratory, China, Email: \{lvbj@mail.sustc.edu.cn, wang.r@sustc.edu.cn\}.

		Lexiang Huang is with Department of Electrical and Electronic Engineering, The Southern University of Scienece and Technology, China, Email: \{huanglx@mail.sustc.edu.cn\}.

Part of this work has been accepted in IEEE ICC 2018 \cite{WANG2017}. We have extended the conference paper by including the learning-based algorithm in Section IV-C, proactive scheduling algorithm in Section V, and more illustrative simulation results. } } 

\begin{document}

\maketitle
\begin{abstract}
In this paper, downlink transmission scheduling of popular files is optimized with the assistance of wireless cache nodes. Specifically, the requests of each file, which is further divided into a number of segments, are modeled as a Poisson point process within its finite lifetime. Two downlink transmission modes are considered: (1) the base station reactively multicasts the file segments to the requesting users and selected cache nodes; (2) the base station proactively multicasts some file segments to the selected cache nodes without requests. The cache nodes with decoded file segments can help to offload the traffic  via other spectrum. Without the proactive multicast, we formulate the downlink transmission resource minimization as a dynamic programming problem with random stage number, which can be approximated via a finite-horizon Markov decision process (MDP) with fixed stage number. To address the prohibitively huge state space, we propose a low-complexity scheduling policy by linearly approximating the value functions of the MDP, where the bound on the approximation error is derived. Moreover, we propose a learning-based algorithm to evaluate the approximated value functions for unknown geographical distribution of requesting users. Finally, given the above reactive multicast policy, a proactive multicast policy is introduced to exploit the temporal diversity of shadowing effect. It is shown by simulation that the proposed low-complexity reactive multicast policy can significantly reduce the resource consumption at the base station, and the proactive multicast policy can further improve the performance.
\end{abstract}

\section{introduction}\label{sec:intro}
Caching is a promising technology to improve the transmission efficiency of wireless networks by exploiting the multiple transmissions of popular files \cite{SalmanAvestimehr2014,Vincent2013}. In this paper, we consider a flexible deployment scenario where there is no wired connection or dedicated spectrum between the base station (BS) and cache nodes. Thus the cache nodes receive popular files via downlink multicast, sharing the same transmission resources as ordinary users. Moreover, the timeliness of popular files, as mentioned in \cite{Leconte2016}, is also considered in transmission design.


\subsection{Related Works}

There have been a number of works on the optimization of file placement with limited cache size.  For example, it is shown in\cite{Caire2012,Molisch2013}, that cache nodes should cache the most popular files if each user has only access to exactly one cache node. The papers \cite{W.Choi2016,K.B.Huang2017} showed that caching files randomly with optimized probabilities is better than storing the most popular files when each user can be served by multiple cache nodes. In \cite{JunZhang2016}, the authors proposed a mobility-aware file placement policy to improve data offloading rate. Moreover, there are also some works on the coded caching scheme design to exploit the multicast transmissions \cite{coded_cache_1,M.Tao2017}.  With cached files, the authors in \cite{Y.Cui2016} formulated the joint minimization of the average delay and power consumption at the BS as a stochastic optimization problem, and the fetching costs are  added  into the  cost function in \cite{Y.Cui2017}. The authors in \cite{Leconte2016} designed a dynamic file placement algorithm  via timely estimation of file content popularity. In most of the above works, the cost of file placement at the cache nodes is not taken into consideration, as it is assumed to be completed before the phase of file delivery to the users. For some types of popular files, however, there may not be sufficient time for file placement before users' accesses. For instance, a great number of news clips are posted to the websites in the daytime, and there is no off-peak hours for caching at the cache nodes (file placement). Hence, it is also interesting to consider the joint scheduling of file placement and delivery.

There are also some works on the joint scheduling design of caching and downlink file transmission. For example, a file placement and delivery framework for heterogeneous OFDM networks was investigated in \cite{Ansari2016}, where the small BSs can cache the popular files and the overall throughput was maximized in each frame via a joint scheduling algorithm. In \cite{cui2016gc}, an optimal caching and user association policy was proposed to minimize the latency in a cached-enabled heterogeneous network with wireless backhaul. In the above works, the files are delivered to small BSs via dedicated backhaul links, i.e., there is no resource competition with the downlink transmission. Moreover, with the update of cache status, the relation between the scheduling in different slots should also be exploited, which is not considered in the above works. If there is no dedicated link or period for file placement  at the cache nodes, the file placement and delivery can be simultaneously scheduled in a multicast manner \cite{Cui2017}. This raises an trade-off between the transmission resource consumption and the file placement. For example, if more resource is spent in downlink multicast, files will be cached in more devices, which may save the downlink resource in future transmissions. As a result, a joint optimization of file placement and delivery with shared transmission resource and the consideration of total transmission resource consumption becomes necessary, and the method of dynamic programming can be utilized.

In fact, dynamic programming via Markov decision process (MDP) has been considered in  delay-aware resource allocation of wireless systems\cite{Cui2017,Moghadari2013,CuiTIT
	,cui2010,Wang2013}. For example, the infinite-horizon MDP was used to optimize the cellular uplink \cite{Moghadari2013,cui2010} and downlink transmissions \cite{Cui2017}, and relay networks\cite{Wang2013}, where the average transmission delay is either minimized or constrained. Moreover, low-complexity algorithm design is usually considered in the above works to avoid the curse of dimensionality \cite{Shewhart2011Approximate}. However, the popular files to be buffered at the cache nodes  usually have a finite lifetime, and the infinite-horizon MDP may not be suitable in modeling anymore. Nevertheless, the MDP with finite stages is usually more complicated \cite{FHMDP}. This is because the optimal policy depends not only on the system state but also on the stage index. To our best knowledge, it is still an open issue on the low-complexity algorithm design and analysis with finite-horizon MDP.

\subsection{Our Contributions}

In this paper, we consider the scheduling of downlink file transmission with the assistance of wireless cache nodes. Specifically, the requests of each file is modeled as a Poisson point process (PPP) within its lifetime, and two downlink transmission modes are considered: (1) the BS reactively multicasts file segments to the requesting users and selected cache nodes; (2) the BS proactively multicasts some file segments to the selected cache nodes without requests from users. With the decoded file segments, cache nodes can offload the traffic from the BS and serve the users within their coverage area  via different spectrum from the downlink (e.g., Wi-Fi) as \cite{May2016,Lv2018,Molisch2016}. The main contributions of this work are summarized below:
\begin{itemize}
	\item With reactive multicast only, we formulate the minimization of a weighted sum of multicast transmission energy and symbol number for one file within its lifetime as a dynamic programming problem with random stage number. The problem does not follow the standard forms of MDP, and it is difficult to find the optimal solution. We first propose to approximate and bound it via a finite-horizon MDP with fixed stage number. Then, we propose a novel linear approximation method on the value functions of the MDP so that the exponential complexity (i.e., curse of dimensionality) can be reduced to linear. With the knowledge of spatial distribution of requesting users, the approximated value functions can be calculated via analytical expressions; whereas, a learning algorithm is also introduced to evaluate the approximated value functions if the distribution of requesting users is unknown.
	
	\item The approximation error of the finite-horizon MDP is usually difficult to analyze, we shall shed some light on this open issue in our problem. Specifically, we first derive an tight upper bound on the gap between the true value functions and the approximated ones. Then we further derive an analytical lower bound on the optimal (minimum) average transmission cost at the BS. 
	
	\item Given the above scheduling policy of reactive multicast, a per-stage optimization approach for proactive multicast is proposed to further suppress the average transmission cost at the BS.
\end{itemize}
It is shown by simulation that, compared with the baseline schemes, the proposed low-complexity algorithm based on approximated value functions can significantly reduce the resource consumption at the BS, and the proactive multicasting policy can further improve the performance.


\section{System Model}\label{sec:model}

\begin{figure}[tb]
	\centering
	\includegraphics[height=150pt,width=250pt]{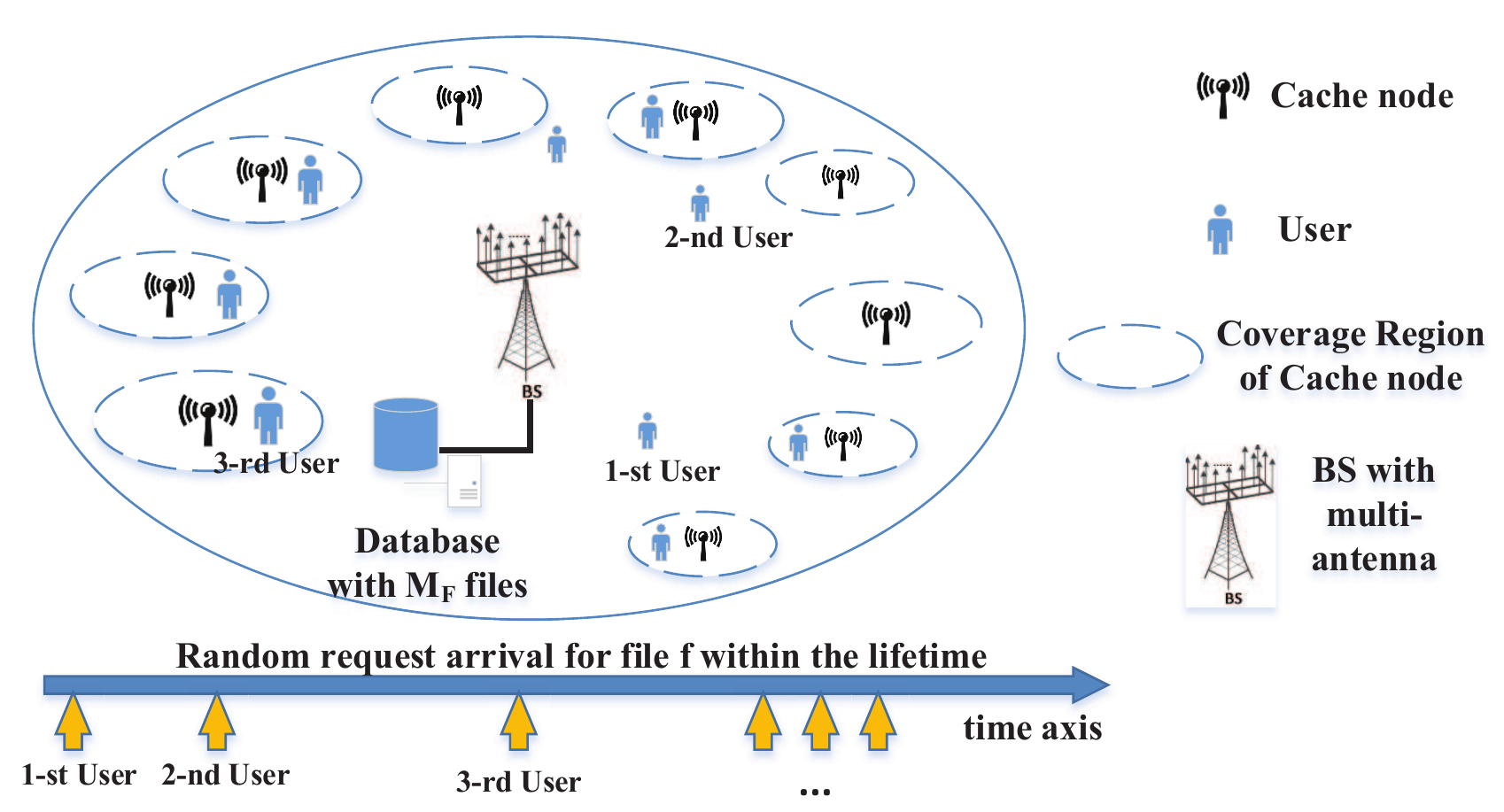}
	\caption{Illustration of network model with one BS and multiple wireless cache nodes.}
	\label{fig:scheme}
\end{figure}

In this section, we introduce the network model for the downlink file transmission with the assistance of wireless cache nodes, and the physical-layer model for the file placement (i.e., transmit files to the cache nodes) and delivery (i.e., transmit files to the requesting users). 

\subsection{Network Model}

As illustrated in Fig. \ref{fig:scheme}, we consider the downlink file transmission in a cell with one BS and $N_C $ single-antenna cache nodes. There are $ N_T $ antennas at the BS. Let $\mathcal{C}\subset \mathbb{R}^2$ be the service area of the cell, $ \mathcal{C}_c $ ($ \forall c=1,2,...,N_C $) be the service region of the $ c $-th cache node and $ \mathcal{C}_0 \triangleq \mathcal{C} -\mathcal{C}^*  $ be the region not served by any cache node, where $ \mathcal{C}^* \triangleq \mathcal{C}_1 \cup \mathcal{C}_2 \cup ... \cup \mathcal{C}_{N_C} $. In this paper, we consider cache node deployment without overlapping, thus assuming $ \mathcal{C}_i \cap  \mathcal{C}_j = \emptyset$ for $\forall i \neq j $. A database with popular files is accessible to the BS. In order to capture the temporal dynamics of files' popularity, it is assumed that the $ f $-th file ($ f=1,2,3,... $) is accessible at the database since the time instance $ t_f $ and remains popular within a finite lifetime $ T_f $. We consider the delivery of the files for the requests raised within their lifetimes (e.g. the lifetime for the $ f $-th file is $ [t_f,t_f+T_f] $). The files are not considered for caching after their lifetime, as their popularity will drop down. It is assumed that the $ f $-th ($ f=1,2,3,... $) file consists of $ N_f $ segments. Each of them, with $ R_f^I $ information bits equally, is encoded separately. Within the lifetime of each file, the locations of the requesting users are independent and identically distributed (i.i.d.) in the cell according to certain spatial distribution with probability density function $ \mathcal{F} : \mathcal{A}\rightarrow [0,1]$, $ \forall \mathcal{A} \subset \mathcal{C} $. It is assumed that the requesting users' locations are quasi-static during the file transmission.

The requests on the $ f $-th file ($ \forall f $) within its lifetime are modeled as a one-dimensional Poisson point process (PPP) with intensity $ \lambda_f $. Hence the probability mass function (PMF) of the remaining request number from the time instance $ t \in [t_f, t_f+T_f]$ is given by
\begin{equation} \label{eqn:request}
\Pr(\mbox{Request Number} = n) = \frac{(\lambda_f T_{rem})^n}{n!} e^{-\lambda_f T_{rem}},
\end{equation}
where $ T_{rem} = t_f+T_f - t $. 

\begin{Remark}[PPP File Request Model]
	The PPP was widely used to model the random phone call arrivals at an exchange. Alternatively, in most of the existing literature \cite{Web_cache,Molisch2013,cui2016gc}, the popularity of files is characterized by the probability of access. The equivalence between the two models are elaborated below. Suppose that there is one file request in each frame with probability $ \beta  \in [0,1]$. Given one file request arrival, the $ f $-th file is requested with probability $ p_f $. $ \{p_f|\forall f\} $ can follow the Zipf distribution with $ \sum_f p_f = 1 $. Then the probability mass function (PMF)  of the request number of $f$-th file within $ N $ frames is given by
	\begin{eqnarray} \nonumber
	&&\Pr(\mbox{Request Number in $ N $ frames} = n) \nonumber\\
	&=& \left( \!\!\!\begin{array}{c}
	N \\ 
	n
	\end{array}\!\!\! \right) (\beta p_f)^n (1-\beta p_f)^{N-n} \rightarrow \frac{e^{-N\beta p_f}(N\beta p_f)^n}{n!}, \nonumber \\ &\mbox{when}& \quad N \rightarrow +\infty.  \nonumber
	\end{eqnarray}
	Note that $ N \beta p_f $ is analogy to $ \lambda_f T_{rem} $ in (\ref{eqn:request}), the Poisson arrival model in (\ref{eqn:request}) is actually consistent with the file request model in the existing literature. Moreover, the condition of sufficient large $ N $ is satisfied as the lifetime is significantly larger than the frame duration.
\end{Remark}

The files may not be cached at the cache nodes before their lifetime. The BS can either proactively multicast some file segments to some cache nodes without any requests, or reactively multicast the segments of one file to the requesting user and some cache nodes if one request is received. In the remaining of this paper, we shall refer to the proactive transmission from the BS to the cache nodes without requests as {\em proactive multicast}, and the reactive transmission from the BS to the requesting user and the cache nodes as {\em reactive multicast}. The proactive multicast is for the file placement, and the reactive multicast should jointly consider both file placement and delivery. In the file delivery, the segments of the requested file will be delivered from the nearby cache node to the requesting user if they have been cached before, and the remaining segments will be multicasted from the BS. We shall refer to the transmission from the service cache node to the requesting user as {\em device-to-device (D2D) file delivery}. It is assumed that the D2D links can use Wi-Fi, bluetooth, or other air interfaces, which are not in the same spectrum as the downlink transmission \cite{May2016,Molisch2016}. In this paper, we shall minimize the total transmission resource consumption at the BS, including the transmission energy and the number of transmission symbols, by offloading traffics to the cache nodes.

\begin{Remark}[Multi-Transmission Scheduling]
	We consider the cached-enabled downlink file transmission where both file placement and delivery should be joint scheduled. For example, if more transmission symbols and power are scheduled in the reactive multicast of one certain file, more cached nodes are able to decode it, which may suppress the downlink  resource consumption in the following requests of this file. Note that the requests arrive at random locations and time instantaneous, it is a stochastic optimization problem, and it is difficult to determine the parameters for all transmissions (including the number of transmission symbols and power) at the very beginning of each file's lifetime. This issue will be addressed via the method of MDP in this paper.
\end{Remark}

\subsection{Downlink Physical Layer Model}

In either proactive or reactive multicast, the space-time block code (STBC) with full diversity is used at the BS for two reasons: (1) there is no requirement on the channel state information at the transmitter (CSIT); (2) diversity gain can be achieved at all the receivers. In the reactive multicast, we refer to the user, which raises the $ n $-th request on the $ f $-th file, as the $ (f,n) $-th user, and refer to the $ s $-th segment of the $ f $-th file as the $ (f,s) $-th segment. Since the transmission time of one file segment is much larger than the channel coherent time of small-scale fading, it is assumed that the ergodic channel capacity span all possible small-scale channel fading and inter-cell interference can be achieved during one segment transmission. Let $ \rho_{f,n} $ and $ \rho_c $ be the pathloss from the BS to the $ (f,n) $-th user and the $ c $-th cache node respectively, $ \eta_{f,n,s} $ and $ \eta_{f,n,s}^c $ be the corresponding shadowing attenuation in $ n $-th transmission of the $ (f,s) $-th segment, $ P_{f,n,s} $ and $ N_{f,n,s} $  be the downlink transmission power and the number of transmission symbols of the $ s $-th file segment in response to the request of the $ (f,n) $-th user. Following the capacity of full-diversity STBC in \cite{paulraj2003introduction},  the throughput achieved by the $ (f,n) $-th downlink user in the reactive multicast of the $ s $-th segment is given by
\begin{equation}\label{eqn:dl-rate}
	R_{f,n,s} = N_{f,n,s} \mathbb{E}_{\mathbf{h}_{f,n,s},I_{f,n}} \bigg[ \alpha \log_2 \bigg( 1 + \frac{||\mathbf{h}_{f,n,s}||^2 P_{f,n,s}}{N_T (\sigma_z^2+I_{f,n})} \bigg) \bigg],
\end{equation}
where $ \alpha $ is the transmission rate of the adopted full-diversity STBC \footnote{For example, $ \alpha = 1 $ and $ N_T = 2$ for Alamouti code. Moreover, $ \alpha $ is usually less than $ 1 $ for $ N_T>2 $.}, $ \sigma_z^2 $ is the power of noise, $I_{f,n}$ is the interference power from the neighbouring BSs\footnote{The exact value of $I_{f,n}$ depends on the scheduling strategies of the neighbouring co-channel BSs, which leads to complicated multi-cell joint scheduling. In order to decouple the scheduling among multiple cells, $I_{f,n}$ can be estimated by assuming all the interfering BSs are transmitting with peak power. Note that the BSs usually use the peak power to broadcast the control information at the head of each frame. One simple approach to measure $I_{f,n}$ is to schedule a few quiet symbols in the frame head, where the service BS does not transmit any signal and the inter-cell interference at the frame head can be measured at the receivers.}, $ \mathbf{h}_{f,n,s} $ is the i.i.d. channel vector from the BS to the requesting user. Each element of $ \mathbf{h}_{f,n,s} $ is complex Gaussian distributed with zero mean and variance $ \rho_{f,n}\eta_{f,n,s} $. As a remark note that the transmission of one segment may consume a large number of frames, and the channel vector $ \mathbf{h}_{f,n,s} $ can be different from frame to frame. Since we consider the ergodic channel capacity, the randomness in small-scale fading is averaged (similar to  \cite{Heath2011,Yang2016Analysis,peter2007})
. As a result, the $ (f,n) $-th user can decode the $ s $-th segment only when $ R_{f,n,s}  \geq R_f^I$. Simultaneously in the reactive multicast, the throughput from the BS to the $ c $-th cache node is given by
\begin{equation}\label{eqn:dl-cache}
	R_{f,n,s}^c = N_{f,n,s} \mathbb{E}_{\mathbf{h}_{f,n,s}^c,I_c} \bigg[ \alpha \log_2 \bigg( 1 + \frac{||\mathbf{h}_{f,n,s}^c||^2 P_{f,n,s}}{N_T (\sigma_z^2+I_c)} \bigg) \bigg],
\end{equation}
where $ \mathbf{h}_{f,n,s}^c $ is the i.i.d. channel vector from the BS to $ c $-th cache node, 	$ I_{c} $ is the interference power from the neighbouring BSs\footnote{$I_c$ can be estimated in a similar way to $I_{f,n}$.}. Each element of $ \mathbf{h}_{f,n,s}^c $ is complex Gaussian distributed with zero mean and variance $ \rho_{c}\eta_{f,n,s}^c $. The $ c $-th cache node can decode the $ (f,s) $-th segment only when $R_{f,n,s}^c  \geq R_f^I$. The throughputs of (\ref{eqn:dl-rate}) and (\ref{eqn:dl-cache}) depend on the pathloss and shadowing of the corresponding links. Hence in the reactive multicast of the $ s $-th segment, the requesting users and cache nodes can decode the segment after receiving different numbers of multicast symbols. By adjusting $ P_{f,n,s} $ and $ N_{f,n,s} $ in physical layer, the BS can control the set of receiving cache nodes. In the next section, we shall formulate the optimization of $ P_{f,n,s} $ and $ N_{f,n,s} $ ($ \forall f,s $) as an MDP with reactive multicast policy.

If periodic proactive multicast is allowed, let $ \eta^c_{k} $ be the shadowing attenuation from the BS to the $ c $-th cache node in the $ k $-th proactive multicast, $ P_k $ and $ N_k $ be the corresponding downlink transmission power and the number of transmission symbols. The throughput achieved by the $ c $-th cache node is given by
\begin{equation}\label{eqn:proactive}
R_{k}^c = N_{k} \mathbb{E}_{\mathbf{h}_{k}^c,I_c} \left[ \alpha \log_2 \left( 1 + \frac{||\mathbf{h}_{k}^c||^2 P_{k}}{N_T (\sigma^2_z+I_c)} \right) \right],
\end{equation}
where $ \mathbf{h}_{k}^c $ is the i.i.d. channel vector from the BS to $ c $-th cache node. Each element of $ \mathbf{h}_{k}^c $ is complex Gaussian distributed with zero mean and variance $ \rho_{c}\eta_{k}^c $. The selected file segment can be decoded at the $ c $-th cache node when $R_{k}^c  \geq R_f^I$. By adjusting $ P_{k} $ and $ N_{k} $ in physical layer, the BS can control the set of receiving cache nodes in proactive multicast. In Section \ref{sec:proactive}, the allocation of $ P_{k} $ and $ N_k $ will be considered in proactive multicast policy.

In both reactive and proactive multicasts, it is assumed that the downlink shadowing effect is quasi-static during the transmission period of one file segment, and different segment transmissions may experience different shadowing attenuations. This model could characterize the large file transmission. For example, the playback time of videos may be several minutes, which is larger than the coherent time of shadowing effect. 

In this paper, we shall address the joint scheduling of proactive and reactive multicasts by two steps. In the following Section \ref{sec:opt} and \ref{sec:approximation}, we first consider low-complexity sub-optimal scheduling designs for reactive multicast. Based on the established scheduling framework, the joint consideration of both proactive and reactive multicasts is addressed in Section \ref{sec:proactive}.

\begin{Remark}[Segment-Level Scheduling]\label{rem:two-time-scale}
	There is a two-time-scale scheduling structure in our problem. Take the reactive multicast as an example. $ N_{f,n,s} $ multicast symbols for the $ (f,s) $-th file segment should be scheduled in a large number of frames, say from the $ k $-th frame to the $ (k+m-1) $-th frame. Let $ M_i $ be the number of scheduled symbols in the $ i $-th frame ($ k \leq i \leq k+m-1 $). $ N_{f,n,s} $ and $ \{M_i|  k \leq i \leq k+m-1 \} $ can be referred to as the segment-level and frame-level parameters respectively. Their relation is $ \sum_{i=k}^{k+m-1} M_i = N_{f,n,s}$. 
	
	Due to the fixed frame size, the scheduling of $ \{M_i|  k \leq i \leq k+m-1 \} $ should jointly consider all the active unicasts, multicasts and broadcasts of the BS, i.e., they should be constrained with other transmissions. On the other hand, we assume there is no buffer overflow at the BS and $ N_{f,n,s} $ multicast symbols will always be transmitted, i.e., there is no constraint on $ N_{f,n,s} $ or on the maximum number of frames to finish the segment multicast. In this paper, we focus on the optimization in the segment level, and the scheduling in the frame level is outside the scope of this paper. However, it should be mentioned that given $ N_{f,n,s} $ in segment level, the scheduling in frame level can affect the transmission delay of the file segment. For example, if $ M_i $ is small, larger $ m $ is required to finish the multicast. 
\end{Remark}

\section{Finite-Horizon MDP Formulation for Reactive Multicast} \label{sec:opt}

In this section, the scheduling design for reactive multicast is first formulated as an dynamic programming problem. However, the optimal solution is difficult to obtain due to the random stage number and continuous state space. Hence, a solvable finite-horizon MDP with a fixed number of stages is introduced to approximate the dynamic programming problem.

\subsection{Dynamic Programming Problem Formulation} 

Without proactive multicast, the system state and the scheduling policy are defined as follows.

\begin{Definition}[System State]
	When receiving the $ n $-th request on the $ f $-th file, the system status is uniquely specified by the following set of parameters $ U_{f,n} \triangleq  \left\{\mathcal{B}_{k,s}^c,\rho_{f,n} , \eta_{f,n,s}, \eta_{f,n,s}^c| \forall k, c=1,...,N_C; s=1,...,N_f  \right\} = S_{f,n} \cup \left\{\mathcal{B}_{k,s}^c| \forall k \neq f, c, s \right\},$ where the $ S_{f,n} = \left\{\mathcal{B}_{f,s}^c,\rho_{f,n} , \eta_{f,n,s}, \eta_{f,n,s}^c| \forall c,s  \right\}  $, $ \mathcal{B}_{k,s}^c =1 $ means that the $ (k,s) $-th segment has been successfully decoded and stored at the $ c $-th cache node and $ \mathcal{B}_{k,s}^c =0 $ means otherwise. $ U_{f,n} $ and $ S_{f,n} $ are referred to as the global and per-file system state of the $ (f,n) $-th reactive multicast, respectively.
\end{Definition}

\begin{Definition}[Reactive Multicast Policy]
	Let $ J_{f,n} $ be the set of segments which should be transmitted to the $ (f,n) $-th user via downlink, i.e., 
	\begin{align}
	J_{f,n}=
	\begin{cases}
	\bigcup\limits_{\{s|\mathcal{B}_{f,s}^{c}=0\}}\{s\}, & \mbox{when } \mathbf{l}_{f,n} \in \mathcal{C}_{c} \ \mbox{and} \ c \neq 0
	\cr \{1,2,...,N_f\}, & \mbox{when } \mathbf{l}_{f,n} \in \mathcal{C}_{0},
	\end{cases}
	\end{align}
	where $ \mathbf{l}_{f,n} $ is the location of the $ (f,n) $-th user. Let $ T_{f,n}^k $ be the remaining lifetime of the $ k $-th file when receiving the $ n $-th file request on the $ f $-th file. The scheduling policy $ \Omega_{f,n} $ ($ \forall f,n $) is a mapping from system state $ U_{f,n} $ and the remaining lifetimes of all the files $ \{T_{f,n}^k|\forall k\} $ to the transmission parameters $ (P_{f,n,s}, N_{f,n,s}) $ ($ \forall s \in J_{f,n} $) and the set of receiving cache nodes for multicast $ \mathbf{c}_{f,n,s}$($ \forall s \in J_{f,n} $), i.e. $\Omega_{f,n}( U_{f,n}, \{T_{f,n}^k|\forall k\})= \{(P_{f,n,s}, N_{f,n,s}),\mathbf{c}_{f,n,s}  | \forall s \in J_{f,n}\}. $ Meanwhile, the following constraints should be satisfied.
	\begin{itemize}
		\item Successful decoding  of each file segment at the requesting user: 
		\begin{align}\label{constraint:decoding}
			R_{f,n,s} \geq R_f^I, \ \ \forall s \in J_{f,n}.
		\end{align}
		\item Successful decoding  of each file segment at the selected cache nodes: 
		\begin{align}\label{constraint:cache_decoding}
			R_{f,n,s}^c \geq R^I_f, \ \ \forall c \in \mathbf{c}_{f,n,s},  s \in J_{f,n}.
		\end{align}	
		\item  Peak power constraint:
	\begin{align}\label{constraint:power_constraint}
	P_{f,n,s} \leq P_B, \ \ \forall s \in J_{f,n},n.
    \end{align}
    where $P_B$ is a instantaneous power constraint at the BS.

	\end{itemize}

\end{Definition}

As mentioned in Remark \ref{rem:two-time-scale}, we shall minimize the total transmission resource consumption on the popular files at the BS by offloading traffics to the cache nodes, so that more transmission resource can be spared for other downlink data or uplink transmission. Let $ \mathcal{C}_{f,n}^s = \cup_{\{\forall i | \mathcal{B}_{f,s}^i=1\}} \mathcal{C}_i $ be the area where the requesting users are able to receive the $ (f,s) $-th file segment from cache nodes,  we use the following cost function to measure the weighted sum of the transmission energy ($ P_{f,n,s}N_{f,n,s} $) and the number of transmission symbols ($ N_{f,n,s} $) of the BS  spent on the $ s $-th segment for the $ (f,n) $-th user.
\begin{equation}\nonumber
g_{f,n,s} \!(\!P_{f,n,s}, N_{f,n,s}\!) = I\!(\!\mathbf{l}_{f,n} \notin \mathcal{C}_{f,n}^s\!) \!\times\! (\!P_{f,n,s}N_{f,n,s} + w N_{f,n,s}\!),
\end{equation}
where $ w $ is the weight on the number of transmission symbols, and $ I(\cdot) $ is the indicator function.

\begin{Remark}[Trade-off between transmission time and energy] If the minimization objective is the total number of transmission symbols spent on one file, the BS will always use the peak power, which might not be energy-efficient. When the traffic load of the BS is not heavy, saving energy is an important design criterion of resource allocation. On the other hand, if  the minimization objective is the total transmission energy spent in one file, it is possible that the BS will use a small power in downlink multicast, which may occupy a large amount of transmission symbols. Thus it is not suitable for heavy traffic load. As a result, we choose a linear combination of both metrics, where the weight on the number of transmission symbols ($ w $) can be chosen according to the traffic load.
\end{Remark}

 Hence the average cost spent on the overall lifetime of the $ f $-th file is given by
\begin{align}\nonumber
&\overline{g}_{f} \left( \{\Omega_{f,n}|\forall n\} \right) \nonumber \\ 
&= \sum_N \mathbb{E}_{\eta,\rho,\mathcal{T}} \left[ \frac{(\lambda_f T_f)^N}{N!} e^{-\lambda_f T_f} \sum_{n=1}^{N} \sum_{s=1}^{N_f} g_{f,n,s}(\Omega_{f,n})\bigg|S_{f,1}\right],
\end{align}
where the expectation is taken over all possible large-scale channel fading (including the shadowing effect $ \eta $ and requesting user's pathloss $ \rho $) and the remaining lifetimes after each file request $\mathcal{T}=\{T_{f,n}^f|\forall n\}$. The summation on $ N $ is to take expectation on the random number of requests as elaborated in (\ref{eqn:request}). Hence the overall system cost function on popular files is $
\overline{G}(\{\Omega_{f,n}|\forall f,n\}) =  \sum_{f=1}^{F} \overline{g}_{f} \left( \{\Omega_{f,n}|\forall n\} \right),$
where $ F $ is the total number of popular files considered in the optimization. The system optimization problem can be written as

\begin{Problem}[Overall System Optimization] \label{prob:overall}
\begin{eqnarray}
&\min\limits_{\{\Omega_{f,n}|\forall f,n\}} &\overline{G}(\{\Omega_{f,n}|\forall f,n\}) \nonumber\\
&s.t.&  \mbox{Constraints in }(\ref{constraint:decoding}-\ref{constraint:power_constraint}).\nonumber
\end{eqnarray}
\end{Problem}

In this paper, we consider the delivery of popular files, and there is sufficient cache space in each cache node to save the popular files in their lifetime. In fact, since all the cached files are received from downlink and all the files have finite lifetimes, the cache size may not be the critical bottleneck of the cache-enabled network considered in this paper.  For example, suppose that one BS is transmitting popular files with overall data rate of $ 100 $ Mbps, and the lifetime of each file is $ 24 $ hours. Then the maximum required cache capacity in one cache node is around $ 1 $ Tera bytes, which is mild. Therefore, the cache size limitation is ignored in this paper. Moreover, as mentioned in Remark \ref{rem:two-time-scale}, there is no constrain on $ \{N_{f,n,s}|\forall f,n,s\} $ (no transmission buffer overflow at the BS). Then the above Problem 1 can be further decoupled into the following per-file sub-problems.

\begin{Problem}[Optimization on the $ f $-th File]\label{prob:main}
\begin{eqnarray}
\overline{g}_f^* = &\min\limits_{\{\Omega_{f,n}|\forall n\}} &\overline{g}_f(\{\Omega_{f,n}|\forall n\}) \nonumber\\
&s.t.&  \mbox{Constraints in }(\ref{constraint:decoding}-\ref{constraint:power_constraint}).\label{eqn:cont}
\end{eqnarray}
\end{Problem}

\begin{figure*}
\begin{align}\label{eqn:W}
W(S_{f},T)=\min\limits_{\{\Omega_{f,k}|\forall k=n+1,...\}} \sum_N \mathbb{E}_{\eta,\rho,\mathcal{T}} \left[ \frac{(\lambda_f T)^N}{N!} e^{-\lambda_f T} \sum_{n=1}^{N} \sum_{s=1}^{N_f} g_{f,n,s}(\Omega_{f,n}) \bigg|S_{f}\right]
\end{align}	
\begin{align}\label{eqn:opt-W}
\Omega_{f,n}^{\dagger}(S_{f,n},T_{f,n}^f)  =\arg \min_{\Omega_{f,n}}\bigg \{ \sum_{s} g_{f,n,s}(\Omega_{f,n})  +  \mathbb{E}_{S_{f,n+1}}[W(S_{f,n+1},T_{f,n}^f)|S_{f,n} ] \bigg\}
\end{align}
\begin{align}
V_{N_R-n+1}(S_{f,n})=\min_{\Omega_{f,n}(S_{f,n})}\bigg\{ \sum_{s} g_{f,n,s}(\Omega_{f,n})+\sum\limits_{S_{f,n+1}}{V_{N_R-n}(S_{f,n+1})\Pr(S_{f,n+1}|S_{f,n},\Omega_{f,n})} \bigg\}, \forall S_{f,n},\label{eqn:bellman-fix}
\end{align}
\begin{eqnarray}\label{eqn:bellman-reduce}
\widetilde{V}_{N_R-n+1}(\widetilde{S}_{f,n})=\min_{\Omega_{f,n}(\widetilde{S}_{f,n})}\mathbb{E}_{\eta,\rho} \bigg\{ \sum_{s}\! g_{f,n,s}(\Omega_{f,n}) +\sum\limits_{\widetilde{S}_{f,n+1}}{\widetilde{V}_{N_R-n}(\widetilde{S}_{f,n+1})\Pr(\widetilde{S}_{f,n+1}|{S}_{f,n},\Omega_{f,n})}  \bigg\}, \forall \widetilde{S}_{f,n},
\end{eqnarray}
\begin{eqnarray}\label{eqn:bellman-random}
\Omega_{f,n}^{*}(S_{f,n},T_{f,n}^f)=\arg \min_{\Omega_{f,n}} \sum_{s} g_{f,n,s}(\Omega_{f,n}) + \sum\limits_{N,\widetilde{S}_{f,n}} \!\! \!\frac{(\lambda_f T_{f,n}^f)^N}{N!} e^{-\lambda_f T_{f,n}^f}{\widetilde{V}_{N}(\widetilde{S}_{f,n+1})\Pr(\widetilde{S}_{f,n+1}|{S}_{f,n},\Omega_{f,n})}, \forall S_{f,n},T_{f,n}^f,
\end{eqnarray}
\hrulefill	
\end{figure*}

Hence, the scheduling policy for the $ f $-th file $ \{\Omega_{f,n} | \forall n\} $ depends only on the per-file system state $ S_{f,n} $ and the its remaining lifetime $ T_{f,n}^f $, i.e. $$\Omega_{f,n}( S_{f,n}, T_{f,n}^f)= \{(P_{f,n,s}, N_{f,n,s}, \mathbf{c}_{f,n,s}) | \forall s \in J_{f,n}\}.$$ 
In order to solve the Problem \ref{prob:main}, we first define the {\em{cost-to-go function}} $W(S_{f},T)$ as the minimum average cost on the $ f $-th file, given the initial per-file system state $S_{f}$ and remaining lifetime $T$, as \eqref{eqn:W}.
Hence, suppose that the per-file system state and the remaining lifetime for the $ n $-th request of the $ f $-th file are $S_{f,n}$ and $T_{f,n}^f$ respectively, the optimal reactive multicast policy for this file transmission $ \Omega_{f,n}^{\dagger}(S_{f,n},T_{f,n}^f)  $ is given by minimizing the summation of current transmission cost and the minimum average future cost, which is given by \eqref{eqn:opt-W}, where the constraints in (\ref{constraint:decoding}-\ref{constraint:power_constraint}) should be satisfied. If $ (S_{f,n},T_{f,n}^f) $ is treated as the system state, its evolution is Markovian. Notice that the number of requests is random and $ T_{f,n}^f $ is continuous, it is difficult to find the cost-to-go function $ W $ accurately and solve the above optimization problem via the standard solution of MDP. In the following section, we shall propose an approximation approach via an MDP with a fixed number of stages.

\subsection{Approximation of Cost-to-go Function} \label{sec:approx-cost-to-go}

In order to solve Problem \ref{prob:main}, we first introduce the following intermediate MDP problem with fixed $ N_R $ requests (stages) on the $ f $-th file, which is similar to Problem \ref{prob:main} except for the stage number.

\begin{Problem}[Optimization with a Fixed Request Number] \label{prob:fix}
	\begin{eqnarray}
	&\min\limits_{\{\Omega_{f,n}|\forall n\}} &\mathbb{E}_{\eta,\rho} [\sum_{n=1}^{N_R}\sum_{s=1}^{N_f} g_{f,n,s}]\nonumber\\
	&s.t.& \mbox{Constraints in }(\ref{constraint:decoding}-\ref{constraint:power_constraint}), \nonumber
	\end{eqnarray}
	where $ N_R $ is the number of requests on the $ f $-th file.
\end{Problem}

The optimal solution of Problem \ref{prob:fix} can be deduced via the Bellman's equations in \eqref{eqn:bellman-fix}, where $V_{N_R-n+1}(S_{f,n}) $ is the {\em value function} of the $ n $-th stage, and $ S_{f,n+1} $ denotes the next state of the $ f $-th file given the current state $ S_{f,n} $, the constraints in (\ref{constraint:decoding}-\ref{constraint:power_constraint}) shall be satisfied in minimizing the right-hand-side of the above equation. The state transition probability can be written as 
\begin{align} 
&\Pr(S_{f,n+1}|S_{f,n},\Omega_{f,n}) \nonumber \\
=&\Pr({\rho}_{f,n+1})\prod_{\forall s}\Pr({\eta}_{f,n+1,s}) \prod_{\forall c,s}\Pr({\eta}_{f,n+1,s}^c) \nonumber \\
& \times I \bigg[\{\mathcal{B}_{f,s}^c (n+1)|\forall s, c\} \bigg], \nonumber
\end{align}
where $ \mathcal{B}_{f,s}^c(n+1) $ is the cache status after taking the action $ \Omega_{f,n} (S_{f,n}) $, $ I $ is the indicator function. In fact, $ V_{N_R-n+1}(S_{f,n}) $ measures the average remaining cost of the $ f $-th file from the $ n $-th transmission to the $ N_R $-th transmission, given the system state of the $ n $-th stage $ S_{f,n} $. Since the large-scale fading is i.i.d. in each file transmission, the expectation on large-scale fading can be taken on both side of the above equation. Hence we have the following conclusion.

\begin{Lemma}[Bellman's Equation with Reduced Space]\label{Lemma:reduce_space}
The optimal control policy of Problem \ref{prob:fix} is the solution of the Bellman's equation with reduced state space as \eqref{eqn:bellman-reduce}, where $\widetilde{S}_{f,n}= \{\mathcal{B}_{f,s}^c \in S_{f,n}|\forall c,s \}$ denotes the cache state of the $ f $-th file, $\widetilde{V}_{N_R-n}(\widetilde{S}_{f,n+1})=\mathbb{E}_{\eta,\rho}[V_{N_R-n}(S_{f,n+1})] $, and $ \Omega_{f,n}(\widetilde{S}_{f,n}) = \{\Omega_{f,n}(S_{f,n})|\forall \rho_{f,n} , \eta_{f,n,s}, \eta_{f,n,s}^c \} $.
\end{Lemma}

\begin{proof}
Please refer to Appendix A.
\end{proof}

The standard {\em value iteration} can be used to solve the Bellman's equations (\ref{eqn:bellman-reduce}), and obtain the value functions $ \widetilde{V}_{N_R-n+1} $ ($ \forall n$) as in \cite{Bertsekas2000Dynamic}. In the following lemma, the cost-to-go function $ W $ defined in (\ref{eqn:W}) can be lower-bounded via the value functions $ \widetilde{V}_{N_R-n+1} $ ($ \forall n$).

\begin{Lemma}[Lower Bound of Cost-to-Go Function]\label{lem:cost_lower_bound}
	With the value function $ \widetilde{V}_{N_R-n+1}(\widetilde{S}_{f,n}) $ ($ \forall n $), given the per-file system state $ S_{f,n} $, remaining lifetime $ T_{f,n}^f $ and reactive multicast policy $ \Omega_{f,n} $ for the $ n $-th request of the $ f $-th file, the minimum average future cost is lower-bounded as 
\begin{align}
	 &\mathbb{E}_{S_{f,n+1}}\![W(\!S_{f,n+1},T_{f,n}^f\!)|S_{f,n},\Omega_{f,n}] \nonumber\\
	\geq&\!\!\!\!\! \sum\limits_{N,\widetilde{S}_{f,n+1}} \!\! \!\!\!\!\!\frac{(\lambda_f T_{f,n}^f)^N}{N!} e^{-\lambda_f T_{f,n}^f}{\widetilde{V}_{N}(\!\widetilde{S}_{f,n+1}\!)\!\Pr(\widetilde{S}_{f,n+1}|{S}_{f,n},\Omega_{f,n})}.
\end{align}
\end{Lemma}
\begin{proof}
	Please refer to Appendix B.
\end{proof}

Using the above lower bound to approximate the cost-to-go function, the optimal scheduling policy in (\ref{eqn:opt-W}) becomes suboptimal as follows.

where the constraints in (\ref{constraint:decoding}-\ref{constraint:power_constraint}) should be satisfied.

It can be observed from (\ref{eqn:bellman-random}) that the scheduling policy for the $ n $-th transmission can be obtained by minimizing the sum of the current transmission cost $ \sum_{s} g_{f,n,s}(\Omega_{f,n}) $ and the lower bound of average future transmission cost $\sum\limits_{N,\widetilde{S}_{f,n+1}}\!\!\!\!\!\frac{(\lambda_f T_{f,n}^f)^N}{N!} e^{-\lambda_f T_{f,n}^f}{\widetilde{V}_{N}(\widetilde{S}_{f,n+1})\Pr(\widetilde{S}_{f,n+1}|{S}_{f,n},\Omega_{f,n})} $, where the latter depends on the value functions $ \widetilde{V}_{N}(\widetilde{S}_{f,n+1})  $. In fact, the state space of $ \widetilde{S}_{f,n+1} $ is huge, which grows exponentially with respect to the number of cache nodes $ N_C $ and the number of segments $ N_f $. The accurate evaluation of the value functions is computationally prohibitive. In the following section, we shall propose (1) an analytical approximation of the value functions, such that the computation complexity can be essentially reduced; (2) an analytical lower bound on the cost-to-go function $ W $, such that the gap between the proposed sub-optimal policy and the optimal scheduling policy can be bounded.

\begin{figure*}
	\begin{equation}\label{eqn:app-v}
	\widetilde{V}_{N_R-n+1}(\widetilde{S}_{f,n}) \approx\underbrace{ \widetilde{V}_{N_R-n+1}(\widetilde{S}_{f}^*)+  \sum_{\{(i,s)|\forall \mathcal{B}^i_{f,s}(\widetilde{S}_{f,n}) =0\}} \bigg( \widetilde{V}_{N_R-n+1}(\widetilde{S}_{f}^{i,s})-\widetilde{V}_{N_R-n+1}(\widetilde{S}_{f}^*)\bigg)}_{\mbox{denote as }\widehat{V}_{N_R-n+1}(\widetilde{S}_{f,n})},
	\end{equation}
	\hrulefill
\end{figure*}

\section{Low-Complexity Solution via Approximate MDP}\label{sec:approximation}

In this section, we shall propose a novel linear approximation approach on the value function $ \widetilde{V}_{N_R-n+1}$ ($ \forall n $), derive the scheduling policy given the current system state and approximated value functions, and analyze the approximation error. We shall also propose a reinforcement learning algorithm for evaluating the approximated value functions with unknown distribution $ \mathcal{F} $ of the requesting users. 

\subsection{Approximation of Value Function} \label{sub:app}

We first define the notations for the following reference cache states.
\begin{itemize}
	\item $\widetilde{S}_{f}^*= \{\mathcal{B}_{f,s}^{c} =1| \forall c, s\}$ is the cache state of the $f$-th file where all the cache nodes have successfully decoded the whole file.
	
	\item $
	\widetilde{S}_{f}^{i,s}  \{\mathcal{B}_{f,s}^{i} = 0, \mathcal{B}_{f,t}^{j} =1| \forall (j,t) \neq (i,s) \} $ is the cache state of $ f $-th file where  all the cache nodes have successfully decoded the whole file except the $ s $-th segment at the $ i$-th cache node. 
\end{itemize}
Hence, we approximate the value function $ \widetilde{V}_{N_R-n+1}(\widetilde{S}_{f,n}) $ linearly as \eqref{eqn:app-v}, where $\mathcal{B}^i_{f,s}(\widetilde{S}_{f,n})  $ means the parameter of $ \mathcal{B}^i_{f,s} $ in the cache state $ \widetilde{S}_{f,n} $. An example of approximated value function is elaborated below.
\begin{figure}[tb]
	\centering
	\includegraphics[height=90pt,width=220pt]{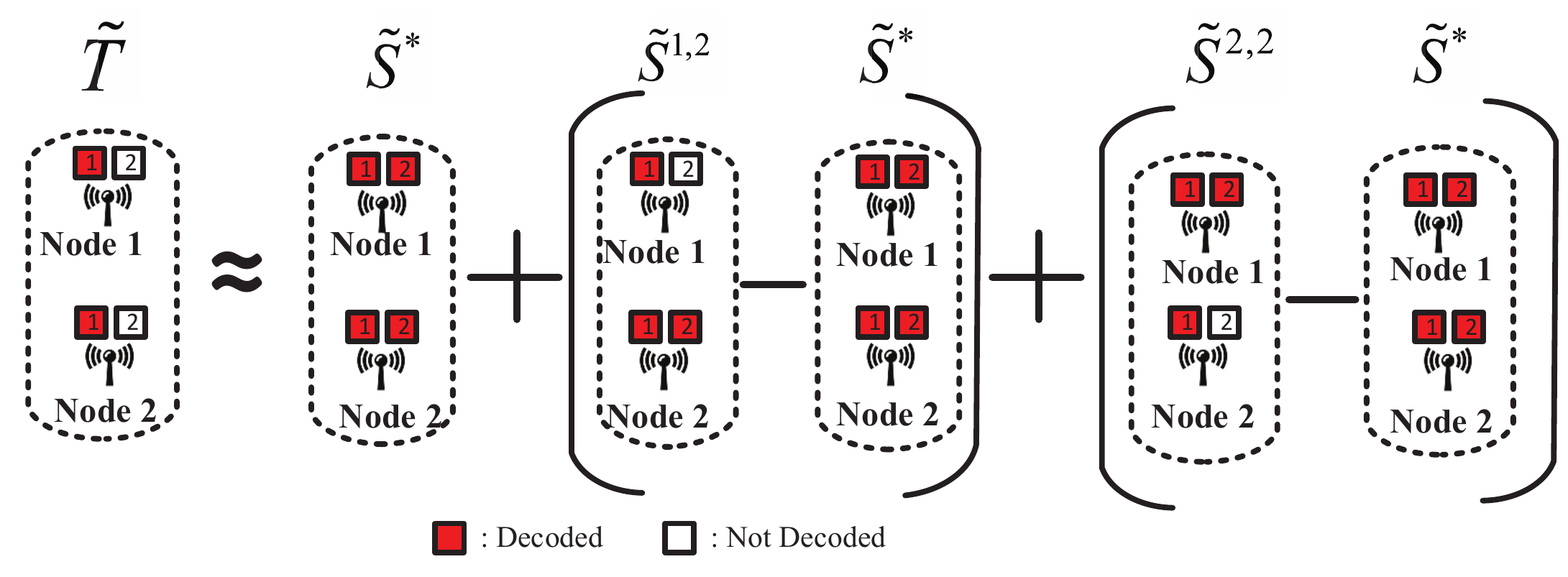}
	\caption{One example of approximated value function.}
	\label{fig:example2}
\end{figure}
\begin{Example}
	An illustrated in Fig. \ref{fig:example2}, there are two cache nodes and the downlink file (say the $ f $-th file) is divided into two segments. For the system state $  \widetilde{T} = [\mathcal{B}^1_{f,1}, \mathcal{B}^1_{f,2}, \mathcal{B}^2_{f,1}, \mathcal{B}^2_{f,2}] = [1, 0, 1, 0] $, the value function on the $ n $-th stage can be approximated as
	\begin{eqnarray}
	\widetilde{V}_{N_R-n+1}(\widetilde{T}) &\approx&	\widehat{V}_{N_R-n+1}(\widetilde{T}) \nonumber\\ &=&\widetilde{V}_{N_R-n+1}(\widetilde{S}_{f}^*) \nonumber\\
	&&+\bigg( \widetilde{V}_{N_R-n+1}(\widetilde{S}_{f}^{1,2})-\widetilde{V}_{N_R-n+1}(\widetilde{S}_{f}^*)\bigg) \nonumber\\
	&&+\bigg( \widetilde{V}_{N_R-n+1}(\widetilde{S}_{f}^{2,2})-\widetilde{V}_{N_R-n+1}(\widetilde{S}_{f}^*)\bigg), \nonumber
	\end{eqnarray}
	where the cache states $ \widetilde{S}_{f}^{1,2} $ and $ \widetilde{S}_{f}^{2,2} $ are illustrated in Fig. \ref{fig:example2}.	In the right hand side of the above approximation, the first term counts the transmission cost for the users outside the coverage region of the cache nodes; the second term approximates the cost on the second segment transmission to the users within the coverage region of the first cache node $ \mathcal{C}_1 $; and the third term approximates the cost on the second segment transmission to the users within the coverage region of the second cache node $ \mathcal{C}_2 $. Note that there is no transmission cost on the first segment for the users within $ \mathcal{C}_1 \cup \mathcal{C}_2 $.
\end{Example}

In order to apply this approximation on all value function, it is necessary to obtain the value of $\widetilde{V}_{N_R-n+1}(\widetilde{S}_{f}^*)$ and $\widetilde{V}_{N_R-n+1}(\widetilde{S}_{f}^{i,s})$ for all $ n,i $, and $ s $ via (\ref{eqn:bellman-reduce}). In the following, we provide the analytically expressions for them with the distribution knowledge of the requesting users. Moreover, an online learning algorithm is proposed in Section \ref{sub:learning} for the evaluation of $\widetilde{V}_{N_R-n+1}(\widetilde{S}_{f}^*)$ and $\widetilde{V}_{N_R-n+1}(\widetilde{S}_{f}^{i,s})$ with unknown spatial distribution of requesting users.

\subsubsection{Evaluation of $ \widetilde{V}_{N_R-n+1}(\widetilde{S}_{f}^*) $}
\begin{figure*}
 \begin{align}\label{eqn:v_default one fragment}
 \widetilde{V}_{N_R-n+1}(\widetilde{S}_{f}^{i,s})=&\mathbb{E}_{\eta,\rho}[\overline{G}_n^1(\widetilde{S}_f^{i,s})|R_{f,n,s} \leq R_{f,n,s}^i]\Pr(R_{f,n,s} \leq R_{f,n,s}^i)\nonumber \\ &+ \mathbb{E}_{\eta,\rho}[\min \{\overline{G}_n^2(\widetilde{S}_f^{i,s}), \overline{G}_n^3(\widetilde{S}_f^{i,s})\}|R_{f,n,s} > R_{f,n,s}^i]\Pr(R_{f,n,s}> R_{f,n,s}^i).
 \end{align}
 \begin{align}
 \widetilde{V}_{N_R-n+1}(\widetilde{S}_{k,n})  \approx \frac{N_k R_k^I}{N_f R_f^I}\widetilde{V}_{N_R-n+1}(\widetilde{S}_{f}^*)
 + \sum_{\forall i} \frac{R_k^I}{R_f^I}\bigg( \!\sum_{\forall s} I[\mathcal{B}^i_{k,s}(\widetilde{S}_{k,n})=0] \! \bigg)  \times\bigg( \! \widetilde{V}_{N_R-n+1}(\widetilde{S}_{f}^{i,1})-\widetilde{V}_{N_R-n+1}(\widetilde{S}_{f}^*)\!\bigg).\label{eqn:value_extend}
 \end{align}
 	\hrulefill
\end{figure*}

Note that the cache state $ \widetilde{S}_{f}^* $ represents the situation that all the cache nodes have already decoded the $ f $-th file, the purpose of downlink transmission is only to make sure that the requesting users, which are outside of the coverage region of any cache node, can decode the downlink file. Hence it is clear that
\begin{align}\nonumber
\widetilde{V}_{N_R-n+1}\!(\!\widetilde{S}_{f}^*\!) \!=& (N_R-n+1)\Pr(\mathbf{l}_{f,n} \notin \mathcal{C}_{f,n}^s)\nonumber \\ 
&\!\!\!\!\times\mathbb{E}_{\rho,\eta}\! \bigg[\!\!\sum_s \!\!\!\min \limits_{P_{f,n,s} \atop N_{f,n,s}} \!\!\! P_{f,n,s}N_{f,n,s} \!\!+\! w N_{f,n,s} \! \bigg| \mathbf{l}_{f,n} \!\!\notin\!\! \mathcal{C}_{f,n}^s\!\!\bigg] \nonumber \\
&s.t. \quad \mbox{Constraints in }(\ref{constraint:decoding}-\ref{constraint:power_constraint}).\nonumber
\end{align}
The above value function can be calculated with analytical expression, which is elaborated below.

\begin{Lemma}\label{Lemma:Segment} In the high SINR regime,
the analytical expression of the value function $ \widetilde{V}_{N_R-n+1}(\widetilde{S}_{f}^*) $ ($n=1,2,\cdots,N_R$) is
\begin{align}\nonumber
\widetilde{V}_{N_R-n+1}\!(\!\widetilde{S}_{f}^*\!) 
=& (N_R-n+1)\Pr(\mathbf{l}_{f,n} \notin \mathcal{C}_{f,n}^s) \nonumber\\
&\!\!\!\!\times\mathbb{E}_{\rho,\eta} \bigg[ \!\!\sum_s \!(\! P_{f,n,s}^*N_{f,n,s}^* \!+\! w N_{f,n,s}^* \!)\! \bigg| \mathbf{l}_{f,n} \!\!\notin\! \mathcal{C}_{f,n}^s \!\!\bigg], \nonumber 
\end{align}
where the optimal power $
P_{f,n,s}^* =\min\{\frac{w} { \mathbb{W}(\frac{2^{\theta}w}{e}) },P_B \}$, the optimal transmission symbol number $
N_{f,n,s}^* =\max\{\frac {R_f^I\ln(2)}{\alpha[   \mathbb{W}(\frac{2^{\theta}w}{e})+1 ]}, \frac{R_f^I}{\alpha[\theta+\log_2(P_B)]}\}$, 
$\theta =	\mathbb{E}_{\mathbf{h}_{f,n,s}} \left[ \log_2 \left(  \frac{||\mathbf{h}_{f,n,s}||^2}{N_T (\sigma^2_z+I_{f,n})}\right) \right] $, 
and $\mathbb{W}(x)$ is the Lambert-W function \cite{W}.
\end{Lemma}

\begin{proof}
	Please refer to Appendix C.	
\end{proof}
\subsubsection{Evaluation of $\widetilde{V}_{N_R-n+1}(\widetilde{S}_{f}^{i,s})$}

Given the cache state $ \widetilde{S}_{f}^{i,s} $ for the $n$-th stage, there are only two possible next cache states $  \widetilde{S}_{f}^{i,s} $ and $  \widetilde{S}_{f}^{*} $ in the $(n+1)$-th stage, which are discussed below.
\begin{itemize}
	\item When $ \rho_{f,n}  \eta_{f,n,s} \leq \rho_i \eta_{f,n,s}^i $, thus $ R_{f,n,s} \leq R_{f,n,s}^i$, the $ i $-th cache node is alway able to decode the $ s $-th file segment give that the transmission constraint (\ref{eqn:cont}) should be satisfied. Thus the next state must be $  \widetilde{S}_{f}^{*} $. In this case, the optimized RHS of (\ref{eqn:bellman-reduce}) is given by
	\begin{eqnarray} \label{eqn:qc}
	\!\!\overline{G}_n^1  (\!\widetilde{S}_f^{i,s}\!)\!\!\!\!\!\!\!&=&\!\!\!\!\!\!\!\!\!\!\!\!\min\limits_{\Omega_{f,n}(\!\widetilde{S}_f^{i,s}\!)} \!\! \sum_{t} \! g_{f,n,t}(\!\widetilde{S}_f^{i,s},\Omega_{f,n}\!) \!+ \!{\widetilde{V}_{N_R-n}(\!\widetilde{S}_{f}^{*}\!)}, \\
	&s.t.&P_{f,n,t}\leq P_B, \forall t \mbox{ and } R_{f,n,t} = {R_f^I},\ \  \forall t. \nonumber
	\end{eqnarray}
	
	\item When $ \rho_{f,n}  \eta_{f,n,s} > \rho_i \eta_{f,n,s}^i $, thus $ R_{f,n,s} > R_{f,n,s}^i$, the BS can choose to deliver the $ s $-th segment to the $ (f,n) $-th user, or both user and the $ i $-th cache node. Hence the optimized  RHS of (\ref{eqn:bellman-reduce}) is given by $
	\mathbb{E}  \bigg\{ \!\min \bigg[ \overline{G}_n^2(\widetilde{S}_f^{i,s}),\overline{G}_n^3(\widetilde{S}_f^{i,s}) \bigg] \!\bigg\},$ 
	where $\overline{G}_n^2$ and $\overline{G}_n^3$ are defined below.
	\begin{eqnarray}\label{eqn:a_u}
	\!\!\overline{G}_n^2\!(\!\widetilde{S}_f^{i,s}\!) \!\!\!\!\!\!\! &=& \!\!\!\!\!\!\!\!\!\!\! \min\limits_{\Omega_f(\!\widetilde{S}_f^{i,s}\!)}\!\! \sum_{t} g_{f,n,t}(\!\widetilde{S}_f^{i,s},\Omega_{f,n}\!) \!\!+\!\!{\widetilde{V}_{N_R-n}(\!\widetilde{S}_{f}^{i,s}\!)}, \\
	 &s.t.&  P_{f,n,t}\leq P_B, \forall t \mbox{ and } R_{f,n,t}={R_f^I}, \  \forall t \nonumber
	\end{eqnarray}
	\begin{eqnarray}\label{eqn:a_i}
	\!\!\overline{G}_n^3(\!\widetilde{S}_f^{i,s}\!)\!\!\!\!\!\!\!&=& \!\!\!\!\!\!\!\!\!\!\!\min\limits_{\Omega_f(\widetilde{S}_f^{i,s})} \!\!\sum_{t} g_{f,n,t}(\!\widetilde{S}_f^{i,s},\Omega_f\!) \!+\! {\widetilde{V}_{N_R-n}(\!\widetilde{S}_{f}^{*}\!)},\\
	&s.t. &P_{f,n,t}\leq P_B, \forall t,  R_{f,n,s}^i \!=\!{ R_f^I}, \nonumber\\ & &R_{f,n,t}\!=\!{R_f^I}, \!  \forall t \neq s. \nonumber
	\end{eqnarray}
\end{itemize}
As a result, the expression of $\widetilde{V}_{N_R-n+1}(\widetilde{S}_{f}^{i,s})$ is summarized by the following lemma.
\begin{Lemma}\label{Lemma:default_one_Segment}
	 The value function  $\widetilde{V}_{N_R-n+1}(\widetilde{S}_{f}^{i,s})$ is given by \eqref{eqn:v_default one fragment}. The asymptotically optimal scheduling parameters for $ \overline{G}_n^1 $ and $ \overline{G}_n^2 $ in high SINR regime are the same as Lemma \ref{Lemma:Segment}. The asymptotically optimal scheduling parameters $\{(P_{f,n,t}^i,N_{f,n,t}^i)|\forall t\}$ for $ \overline{G}_n^3 $ is given by	
	 \begin{align*}
	 &P_{f,n,s}^i=\min \{ \frac{w} { \mathbb{W}(\frac{2^{\theta^i}w}{e}) } ,P_B\},\\&N_{f,n,s}^i=\max \{\frac {R_f^I\ln(2)}{\alpha[   \mathbb{W}(\frac{2^{\theta^i}w}{e})+1 ]},\frac{R_f^I}{\alpha[\theta^i+\log_2(P_B)]} \}, \\
		 & \theta^i=	\mathbb{E}_{\mathbf{h}_{f,n,s}^i} \left[ \log_2 \left(  \frac{||\mathbf{h}_{f,n,s}^i||^2}{N_T (\sigma^2_z+I_{i})}\right) \right],  
	 \end{align*}

	 and $ \forall t\neq s $
	 \begin{align*}
	 &P_{f,n,t}^i =\min \{\frac{w} { \mathbb{W}(\frac{2^{\theta}w}{e}) }, P_B\},\\
	 &N_{f,n,t}^i =\max \{\frac {R_f^I\ln(2)}{\alpha[   \mathbb{W}(\frac{2^{\theta}w}{e})+1 ]},\frac{R_f^I}{\alpha[\theta+\log_2(P_B)]}   \},\\ 
	 &\theta =	\mathbb{E}_{\mathbf{h}_{f,n,t}} \left[ \log_2 \left(  \frac{||\mathbf{h}_{f,n,t}||^2}{N_T (\sigma^2_z+I_{f,n})}\right) \right].
	 \end{align*}   
\end{Lemma}
\begin{proof}
The proof is similar to that of Lemma \ref{Lemma:Segment}, and it is omitted here.	
\end{proof}

Hence, it is clear that $\widetilde{V}_{N_R-n+1}(\widetilde{S}_f^{i,s}) = \widetilde{V}_{N_R-n+1}(\widetilde{S}_f^{i,t}), \forall s \neq t$. With the distribution knowledge of large-scale fading, the value functions $\widetilde{V}_{N_R-n+1}(\widetilde{S}_{f}^*)$ and $\widetilde{V}_{N_R-n+1}(\widetilde{S}_{f}^{i,s})$ can be calculated according to above analytical expressions. Moreover, although different files may consist of different number of segments or segment size, the calculation of $\widetilde{V}_{N_R-n+1}(\widetilde{S}_{f}^*)$ and $\widetilde{V}_{N_R-n+1}(\widetilde{S}_{f}^{i,s})$ on one file can be easily extended to the other files. For example, given the cache state $\widetilde{S}_{k,n}$ of the $ k $-th file ($ \forall k \neq f $), the value functions approximation, denoted as $\widetilde{V}_{N_R-n+1}(\widetilde{S}_{k,n})$, can be calculated via $\widetilde{V}_{N_R-n+1}(\widetilde{S}_{f}^*)$ and $\widetilde{V}_{N_R-n+1}(\widetilde{S}_{f}^{i,1})(\forall i)$ for the f-th file as \eqref{eqn:value_extend}.

\subsection{Reactive Multicast Policy } \label{sub:low}

With $\widetilde{V}_{N_R-n+1}(\widetilde{S}_{f}^*)$ and $\widetilde{V}_{N_R-n+1}(\widetilde{S}_{f}^{i,s})$, the value function for arbitrary system state can be approximated via (\ref{eqn:app-v}). Hence the reactive multicast policy, denoted as $ \Omega_{f,n}^{*}({S}_{f,n}, T_{f,n}^f) $, can be obtained. Moreover, as $ \widetilde{V}_{N}(\widetilde{S}_{f,n+1}) $ can be decoupled for each segment, the  optimization problem (\ref{eqn:bellman-random}) can be also decoupled for each segment. Specifically, for the $ s $-th segment ($ \forall s \in J_{f,n}$), the solution of (\ref{eqn:bellman-random}) can be obtained by solving the following problem.
\begin{Problem}[Optimization for the $ s $-th Segment]\label{prob:online-s}
	\begin{eqnarray}
	\{P_{f,n,s}^{*}, N_{f,n,s}^{*}\} \!\!& \!=\! & \!\!\!\!\arg\min g_{f,n,s}({S}_{f,n},\Omega_{f,n}) \nonumber\\
	 &&\!+\! \sum\limits_{N}\bigg\{\frac{(\lambda_f T_{f,n}^f)^N}{N!} e^{-\lambda_f T_{f,n}^f}  \nonumber\\ &&\times \!\!\!\!\!\!\!\!\!\!\sum\limits_{\{i|\forall \mathcal{B}^i_{f,s}(\widetilde{S}_{f,n+1})=0\}}  \!\!\!\! \!\!\!\!\!\! \!\!\! \left[ \widetilde{V}_{N}(\widetilde{S}_{f}^{i,s}) \!-\! \widetilde{V}_{N}(\widetilde{S}_{f}^{*}) \right]\bigg\}\nonumber\\
	&s.t.& \mbox{Constraints in }(\ref{constraint:decoding}-\ref{constraint:power_constraint}),\nonumber
	\end{eqnarray}
where $ \mathcal{B}^i_{f,s}(\widetilde{S}_{f,n+1}) $ represents the next cache state for the $ (f,s) $-th segment in  $ i $-th cache node. Note that the set of receiving cache nodes $ \mathbf{c}^*_{f,n,s} $ can be determined from $ (P_{f,n,s}^{*}, N_{f,n,s}^{*}) $.
\end{Problem}

This is an integrated continuous and discrete optimization, its solution algorithm is summarized below.

\begin{Algorithm}[Scheduling with Approximated Value Function]\label{alg:AMDP}
	
Given the system state $ S_{f,n} $, let $ d_1,d_2,.. $ be the indexes of cache nodes, whose large-scale attenuation to the BS in the $ s $-th segment ($ \forall s \in J_{f,n}$) is worse than the $ (f,n) $-th user. Moreover, without loss of generality, it is assumed that $ \rho_{d_1}\eta_{f,n,s}^{d_1} \leq  \rho_{d_2}\eta_{f,n,s}^{d_2} \leq ... \leq \rho_{f,n,s}\eta_{f,n,s}$. The solution of Problem \ref{prob:online-s} can be obtained by the following steps.
\begin{itemize}
	\item For each $ i $, solve the following optimization problem.
	\begin{align}
	Q_{d_i,s}^{*}({S}_{f,n})=& \min\limits_{P_{f,n,s} \atop N_{f,n,s}} g_{f,n,s}({S}_{f,n},\Omega_{f,n}) \nonumber \\
	 &+ \sum\limits_{N}\bigg\{\frac{(\lambda_f T_{f,n}^f)^N}{N!} e^{-\lambda_f T_{f,n}^f}\nonumber \\ 
	 &\times\bigg[\sum\limits_{j=\{d_1,...,d_{i-1}\}} \!\!\!\!\!\!\widetilde{V}_{N}(\widetilde{S}_{f}^{j,s}) - \widetilde{V}_{N}(\widetilde{S}_{f}^{*})\bigg]\bigg\}\nonumber\\
	s.t.& \quad P_{f,n,s}\leq P_B \mbox{ and }  R_{f,n,s}^{d_i} = R_f^I. \nonumber
	\end{align}
	The solution, denoted as $[P_{f,n,s}^{d_i},N_{f,n,s}^{d_i}]$, can be derived similar to Lemma \ref{Lemma:Segment}. Note that $[P_{f,n,s}^{d_i},N_{f,n,s}^{d_i}]$ are the transmission parameters if the file segment can be decoded in the $ d_i $-th cache node. 
	\item Let $
d^{*}=\arg \min\limits_{d_i} Q_{d_i,s}^* $, the solution of Problem \ref{prob:online-s} is then given by $
[P_{f,n,s}^{*},N_{f,n,s}^{*}]=[P_{f,n,s}^{d^*},N_{f,n,s}^{d^*}]$.
	
\end{itemize}

\end{Algorithm}

\subsection{Learning Algorithm for Approximated Value Function}\label{sub:learning}

In Section \ref{sub:app}, the values of $\widetilde{V}_{N_R-n+1}(\widetilde{S}_{f}^*)$ and $\widetilde{V}_{N_R-n+1}(\widetilde{S}_{f}^{i,s})$ are evaluated analytically by assuming that the distribution of the requesting users $ \mathcal{F} $ is known. However in practice, this distribution may be unknown to the BS. In order to address this issue, we propose the following learning-based algorithm to evaluate the value functions  $\widetilde{V}_{N_R-n+1}(\widetilde{S}_{f}^*)$ and $\widetilde{V}_{N_R-n+1}(\widetilde{S}_{f}^{i,s})$ from the historical request arrivals.

\begin{Algorithm}[Reinforcement Learning for Value Functions]\label{alg:learning}
\ \ 
	\begin{itemize}
		\item {\bf Step 1}: Let $ t=0 $. Initialize the value of $\widetilde{V}_{N_R-n+1}(\widetilde{S}_{f}^*)$ and $\widetilde{V}_{N_R-n+1}(\widetilde{S}_{f}^{i,s})$ ($ \forall n,i,s $), and denote them as $\widetilde{V}^t_{N_R-n+1}(\widetilde{S}_{f}^*)$ and $\widetilde{V}^t_{N_R-n+1}(\widetilde{S}_{f}^{i,s})$. This initialization can be done by assuming all the users appear uniformly in the cell coverage, hence the approach in Section \ref{sub:app} can be applied to calculate the initial values.
		
		\item {\bf Step 2}: Let $ t=t+1 $ if there is file request arrival.  Suppose it is the  $ m $-th request on the $ g $-th file, and the location of the requesting user is $ \mathbf{l}_{m,g} $, we have $ \forall i,f,s,n $
		\begin{align*}
		\widetilde{V}_{N_R-n+1}^t(\widetilde{S}_{f}^*) =&\frac{t}{t+1}\widetilde{V}_{N_R-n+1}^{t-1}(\widetilde{S}_{f}^*)\\
		 &+\frac{R_f^I}{(t+1)R_g^I}\bigg\{(N_R\!-\!n\!+\!1)I\!(\mathbf{l}_{g,m} \! \!\notin \mathcal{C})\\
		&\times\sum_s\! (P_{g,m,s}^*N_{g,m,s}^*  + w N_{g,m,s}^*)\bigg\},\nonumber
		\end{align*}
		where $P_{g,m,s}^* =\min \{\frac{w} { \mathbb{W}(\frac{2^{\theta_t}w}{e}) }, P_B\},
		N_{g,m,s}^* =\frac{R_f^I}{\alpha[\theta_t+\log_2(P_{g,m,s}^*)]} $, $
		\theta_t =	\!\mathbb{E}_{\mathbf{h}_{g,m,s}}\!\! \left[ \log_2 \left(  \frac{||\mathbf{h}_{g,m,s}||^2}{N_T (\sigma^2_z+I_{g,m})}\right)\! \right] $.
		\begin{align*}
		\!\!\!\!\widetilde{V}_{N_R\!-\!n\!+\!1}^t\!(\!\widetilde{S}_{f}^{i,s}\!)\!\!=&\frac{t}{t+1}\widetilde{V}_{N_R-n+1}^{t-1}(\widetilde{S}_{f}^{i,s})\\
		&\!\!\!\!\!+\frac{R_f^I}{(\!t\!+\!1\!)\!R_g^I}\bigg\{ 	\overline{G}_n^1(\widetilde{S}_{g}^{i,s}) I(R_{g,m,s} \leq R_{g,m,s}^i)\\
		 &\!\!\!\!\!\!+ \!\min\{	\!\overline{G}_n^2\! (\!\widetilde{S}_{g}^{i,s}\!), 	\overline{G}_n^3\! (\!\widetilde{S}_{g}^{i,s}\!)\!\} I(R_{g,m,s}\!\!>\!\! R_{g,m,s}^i)\!\bigg\}, \nonumber
		\end{align*}
		where $	\overline{G}_n^1 $, $	\overline{G}_n^2$ and $	\overline{G}_n^3$ are defined in (\ref{eqn:qc}), 
		(\ref{eqn:a_u}) and (\ref{eqn:a_i}) respectively.
		
		\item {\bf{Step 3}}: If $\max\{|\widetilde{V}_{N_R-n+1}^t(\!\widetilde{S}_{f}^{i,s}\!)\!-\! \widetilde{V}_{N_R-n+1}^{t-1}(\!\widetilde{S}_{f}^{i,s}\!)|, |\widetilde{V}_{N_R-n+1}^t(\!\widetilde{S}_{f}^{*}\!)\!-\! \widetilde{V}_{N_R-n+1}^{t-1}(\!\widetilde{S}_{f}^{*}\!)| \big| \forall n,i,s\}  $ is greater than one threshold $ \tau $, the algorithm goes to Step 2; otherwise, the algorithm terminates.
	\end{itemize}
\end{Algorithm}

Moreover, we have the following conclusion on the convergence of above learning algorithm.

\begin{Lemma}\label{Lemma:learning}
	The Algorithm \ref{alg:learning} will converge to the true value of $\widetilde{V}_{N_R-n+1}(\widetilde{S}_{f}^*)$ and $\widetilde{V}_{N_R-n+1}(\widetilde{S}_{f}^{i,s})$ ($ \forall f,n,i,s $). Thus $$
	\lim\limits_{t \rightarrow +\infty } \widetilde{V}^t_{N_R-n+1}(\widetilde{S}_{f}^*) = \widetilde{V}_{N_R-n+1}(\widetilde{S}_{f}^*)$$, 	$$\lim\limits_{t \rightarrow +\infty } \widetilde{V}^t_{N_R-n+1}(\widetilde{S}_{f}^{i,s}) = \widetilde{V}_{N_R-n+1}(\widetilde{S}_{f}^{i,s}). $$

\end{Lemma}

\begin{proof}
	Please refer to Appendix D. 
\end{proof}
\begin{figure*}
		\begin{equation}\label{eqn:refined_upper_bound}
		\overline{V}_{N_R-n+1}(\widetilde{S}_{f,n}) = \!\!\! \min_{\Omega_{f,n}(\widetilde{S}_{f,n})} \!\!\! \mathbb{E}_{\eta,\rho}\! \bigg\{\! \sum_{s}\! g_{f,n,s}(\Omega_{f,n}) \!+ \!\!\! \sum\limits_{\widetilde{S}_{f,n+1}} \!\!{\widehat{V}_{N_R-n}(\widetilde{S}_{f,n+1})\Pr(\widetilde{S}_{f,n+1}|\widetilde{S}_{f,n},\Omega_{f,n})} \! \bigg\}.
		\end{equation}
		\hrulefill
\end{figure*}

\subsection{Bounds on Approximated Value function}

In this paper, two approximation steps are proposed to find the sub-optimal and low-complexity reactive multicast policy, i.e.,
\begin{itemize}
	\item {\bf $ W \rightarrow \widetilde{V} $: } Approximate the cost-to-go function $ W $ via a linear combination of value functions of a finite-horizon MDP in Section \ref{sec:approx-cost-to-go}.
	\item {\bf $ \widetilde{V} \rightarrow \widehat{V} $: } Analytically approximate the value function in Section \ref{sub:app}.
\end{itemize}
In this section, we shall provide an analytical upper bound on the approximation error of $ \widetilde{V} \rightarrow \widehat{V} $, and an analytical lower bound on the cost-to-go function $ W $ (Note that the upper bound of $ W $ can be obtained by numerical simulation). First of all, we have the following conclusion on the bounds of the true value functions $ \widetilde{V}  $. 
\begin{Lemma}[Bounds of Value Functions] \label{lem:bound} 
	The upper and lower bounds of  $ \widetilde{V}_{N_R-n+1}(\widetilde{S}_{f,n}) $ ($ \forall f,n $) are provided below.
\begin{equation}\label{eqn:upper}
\widetilde{V}_{N_R-n+1}(\widetilde{S}_{f,n}) \leq \widehat{V}_{N_R-n+1}(\widetilde{S}_{f,n})
\end{equation}
\begin{align}\label{eqn:lower}
\widetilde{V}_{N_R-n+1}(\widetilde{S}_{f,n}) \geq&  \widetilde{V}_{N_R-n+1}(\widetilde{S}_{f}^*) \nonumber \\
&+ \!\!\!\!\!\!\!\!\!\sum_{{\{(i,s)|\forall \mathbf{B}^i_{f,s}(\widetilde{S}_{f,n}) = 0\}}} \bigg( \widetilde{V}_{1}(\widetilde{S}_{f}^{i,s})-\widetilde{V}_{1}(\widetilde{S}_{f}^*)\bigg)
\end{align}	
\end{Lemma}
\begin{proof}
	Please refer to Appendix E.
\end{proof}

Let $ \mathcal{E}_{N_R-n+1}(\widetilde{S}_{f,n}) \triangleq \widehat{V}_{N_R-n+1}(\widetilde{S}_{f,n}) - \widetilde{V}_{N_R-n+1}(\widetilde{S}_{f,n}) $, ($ \forall n,  \widetilde{S}_{f,n}$) be the approximation error of the value functions for arbitrary $ n $ and $ \widetilde{S}_{f,n} $. Replacing $ \widetilde{V}_{N_R-n+1}(\widetilde{S}_{f,n}) $ by (\ref{eqn:lower}), we have
\begin{align*} 
\mathcal{E}_{N_R-n+1}(\!\widetilde{S}_{f,n}\!)\! \leq&  \!\!\!\!\!\!\!\!\!\!\!\!\!\!\!\!\!\sum_{\{(i,s)|\forall \mathcal{B}^i_{f,s}(\widetilde{S}_{f,n}) = 0\}} \!\! 
\bigg\{\!\widetilde{V}_{N_R-n+1}(\!\widetilde{S}_{f}^{i,s}\!)\!-\!\widetilde{V}_{N_R-n+1}(\!\widetilde{S}_{f}^*\!)\\	&-\widetilde{V}_{1}(\widetilde{S}_{f}^{i,s})+\widetilde{V}_{1}(\widetilde{S}_{f}^*)\bigg\}.
\end{align*}
Moreover, a tighter upper bound of the value functions can be obtained numerically.
\begin{Corollary}[Refined Upper Bound of Value Function]\label{Corollary:refined_upper_bound}
	Let $ \overline{V}_{N_R-n+1}(\widetilde{S}_{f,n}) $ ($ \forall f,n $) be the intermediate value function on system state $ \widetilde{S}_{f,n} $ after one-step value iteration based on $ \widehat{V}_{N_R-n}(\widetilde{S}_{f,n+1}) $, is given by \eqref{eqn:refined_upper_bound}.
	Then, $\widetilde{V}_{N_R-n+1}(\widetilde{S}_{f,n}) \leq \overline{V}_{N_R-n+1}(\widetilde{S}_{f,n}) \leq \widehat{V}_{N_R-n+1}(\widetilde{S}_{f,n})$.
\end{Corollary}
\begin{proof}
	Please refer to Appendix F.
\end{proof}
With the knowledge of $ \mathcal{F} $, $ \overline{V}_{N_R-n+1}(\widetilde{S}_{f,n}) $ can be calculated via Monte Carlo simulation. If $ \mathcal{F} $ is not available at the BS, the learning-based approach can be used to evaluate $ \overline{V}_{N_R-n+1}(\widetilde{S}_{f,n}) $. The algorithm is similar to the one in Section \ref{sub:learning}, and it is omitted here due to page limitation. Hence, it is feasible to obtain better approximation of value function $ \widetilde{V}_{N_R-n+1}(\widetilde{S}_{f,n}) $ for any specified stage and per-file system state.

Note that the upper bound of the cost-to-go function $ W(S_{f,n},T_{f,n}^f) $ can be obtained by simulating the file transmission with initial system state $ S_{f,n} $ and lifetime $ T_{f,n}^f $. We introduce the following analytical lower bound.

\begin{Lemma}[Analytical Lower-bound on $ W $]
	Given the initial system state $ S_{f,n+1} $ at the beginning of a remaining lifetime with duration $ T_{f,n}^f $, the minimum transmission cost of the BS, denoted as $ W(S_{f,n+1},T_{f,n}^f) $ is lower-bounded as
\begin{align*}
&W(S_{f,n+1},T_{f,n}^f) \\\geq& \sum_{N_R}\bigg\{\frac{(\lambda_f T_{f,n}^f)^{N_R}}{N_R!} e^{-\lambda_f T_{f,n}^f} \\ 
&\times \bigg[ \widetilde{V}_{N_R}(\widetilde{S}_{f}^*)+\!\!\!\!\!\!\!\!\!\!\!\!\!\!\!\!\!\!
 \sum_{\{(i,s)|\forall \mathbf{B}^i_{f,s}(\widetilde{S}_{f,n+1}) = 0\}} \bigg( \widetilde{V}_{1}(\widetilde{S}_{f}^{i,s})-\widetilde{V}_{1}(\widetilde{S}_{f}^*)\bigg) \bigg]\bigg\}.  
\end{align*}
\end{Lemma}

\begin{proof}
	This lemma is straightforward by combining the conclusions of Lemma \ref{lem:cost_lower_bound} and \ref{lem:bound}.
\end{proof}

\section{Scheduling Algorithm for Proactive Multicast}\label{sec:proactive}

In this section, we propose a heuristic scheduling algorithm of proactive multicast, which could deliver some file segments to the cache nodes with low transmission cost by exploiting the temporal diversity of shadowing effect. We first define the proactive file placement policy.

 \begin{Definition}[Proactive Multicast Policy]
	Suppose that the BS will proactively multicast one file segment in every $ T_p $ seconds. In the $ k $-th proactive transmission opportunity, given the state of each cache node $ \{ \mathcal{B}_{f,s}^c|\forall c, f, s \} $, the shadowing from the BS to each cache nodes $ \{\eta^c_k | \forall c \} $, and the remaining lifetime of each file $ \{ T^k_f| \forall f \} $, the BS should determine the selected file segment $ (f_k, s_k) $ and the downlink transmission parameters $ P_{k} $ and $ N_{k} $ for the selected $ (f_k, s_k) $-th file segment. Thus denote $ S^k =\left[\{ \mathcal{B}_{f,s}^c|\forall c, f, s \},\{\eta^c_k | \forall c \}, \{ T^k_f| \forall f \} \right]$, the proactive multicast policy can be written as $\Omega_{k}( S^k)= [f_k, s_k, P_{k}, N_{k}]. $
\end{Definition}

The joint optimization of reactive multicast $ \{ \Omega_{f,n}|\forall f,n \} $ and proactive multicast $ \{ \Omega_k |\forall k \} $ is complicated as the transmission strategy of different files are coupled. Instead, we use the low-complexity scheduling policy for reactive multicast derived in the previous section, and propose a heuristic proactive multicast to further suppress the overall transmission cost. Specifically, we use $
\widehat{g}_f^k(\widetilde{S}_{f}^k, T_{f}^k) = \sum\limits_{N}\frac{(\lambda_f T_{f}^k)^N}{N!} e^{-\lambda_f T_{f}^k}{\widehat{V}_{N}(\widetilde{S}_{f}^k)} $
to approximate the remaining transmission cost spent on the $ f $-th file without any proactive multicast, where $ \widetilde{S}_{f}^k = \{ \mathcal{B}_{f,s}^c|\forall c, s \} $ is the cache state of the $ f $-th file before the $ k $-th proactive multicast. Hence, $ \Omega_{k}( S^k) $ can be determined as follows.

\begin{Problem}[Heuristic Scheduling for Proactive Multicast] \label{prob:heur}
\begin{eqnarray}
 & \max\limits_{P_k, N_k, f_k, s_k}& \!\!\!\!\!	\widehat{g}_{f_k}^k(\widetilde{S}_{{f_k}}^k, T_{{f_k}}^k) /  \left[ P_k N_k + w  N_k + \widehat{g}_{f_k}^k(\breve{S}_{{f_k}}^k, T_{{f_k}}^k)\right] \nonumber\\
&s.t.& \!\!\!\!\!\!\!\! \widehat{g}_{f_k}^k(\!\widetilde{S}_{{f_k}}^k, T_{{f_k}}^k
\!) /  \left[\! P_k N_k \!+ \!w  N_k \!+\! \widehat{g}_{f_k}^k(\!\breve{S}_{{f_k}}^k, T_{{f_k}}^k\!)\!\right] \!\geq\! \tau^{'}, \nonumber
\end{eqnarray}
where $ \tau^{'} >1 $ is a constant threshold, $ \widetilde{S}_{{f_k}}^k $ and $ \breve{S}_{{f_k}}^k $ are  the system cache state before and after proactive multicast respectively.
\end{Problem}

In the objective of Problem \ref{prob:heur}, the numerator and the denominator are the approximations of the $ f_k $-th file's remaining transmission cost with and without the $ k $-th proactive multicast, respectively. The constraint with $ \tau^{'} >1  $ is due to approximation error. Problem \ref{prob:heur} can be solved via the following algorithm.

\begin{Algorithm}[Proactive Multicast] On each proactive transmission opportunity (say the $ k $-th opportunity), the algorithm to determine the proactive multicast policy $ \Omega_k $ is elaborated below.
	
\begin{itemize}
\item Step 1: For each file segment (say the $ (f,s) $-th one), evaluate 
\begin{eqnarray}
\!\Delta g_{f,s}^k\!\!\!\!\!\!\! &=&\!\!\!\!\!\! \!\!\!\!\!\!\max\limits_{P_{f,s}^k, N_{f,s}^k, \mathcal{A}_f^k}	\frac{\widehat{g}_f^k(\widetilde{S}_{f}^k, T_{f}^k)} {  P_{f,s}^k N_{f,s}^k + w  N_{f,s}^k + \widehat{g}_f^k(\breve{S}_{f}^k(\mathcal{A}_f^k), T_{f}^k)} \nonumber\\
&s.t.& R_k^c \geq R_f^I, \ \forall c \in \mathcal{A}_f^k, \nonumber
\end{eqnarray}
where $ P_{f,s}^k, N_{f,s}^k$ and $ \mathcal{A}_f^k $ represents the transmission power, transmission symbol number and the set of receiving cache nodes, $ \breve{S}_{f}^k(\mathcal{A}_f^k) $ denotes the cache  state where the cache nodes in $ \mathcal{A}_f^k $ have successfully decoded the $ (f,s) $-th segment given the previous state $ \widetilde{S}_{f}^k $. The solution of the above optimization problem can be obtained by minimizing the denominator, which is similar to that of Problem \ref{prob:online-s}. Hence it is omitted here.

\item Step 2: The $ (f_k,s_k) $-th segment is chosen when the following two conditions are satisfied: 
\begin{itemize}
	\item $ (f_k,s_k)  = \arg\max\limits_{(f,s)} \Delta g_{f,s}^k $ and $ \Delta g_{f_k,s_k}^k \geq \tau^{'} $.
\end{itemize}
\end{itemize}
\end{Algorithm}

\section{Simulation}\label{sec:sim}

\begin{figure} 
	\centering 
	\subfigure[$ N_C = 20 $]{ \label{fig:20caches}  \includegraphics[width=3in]{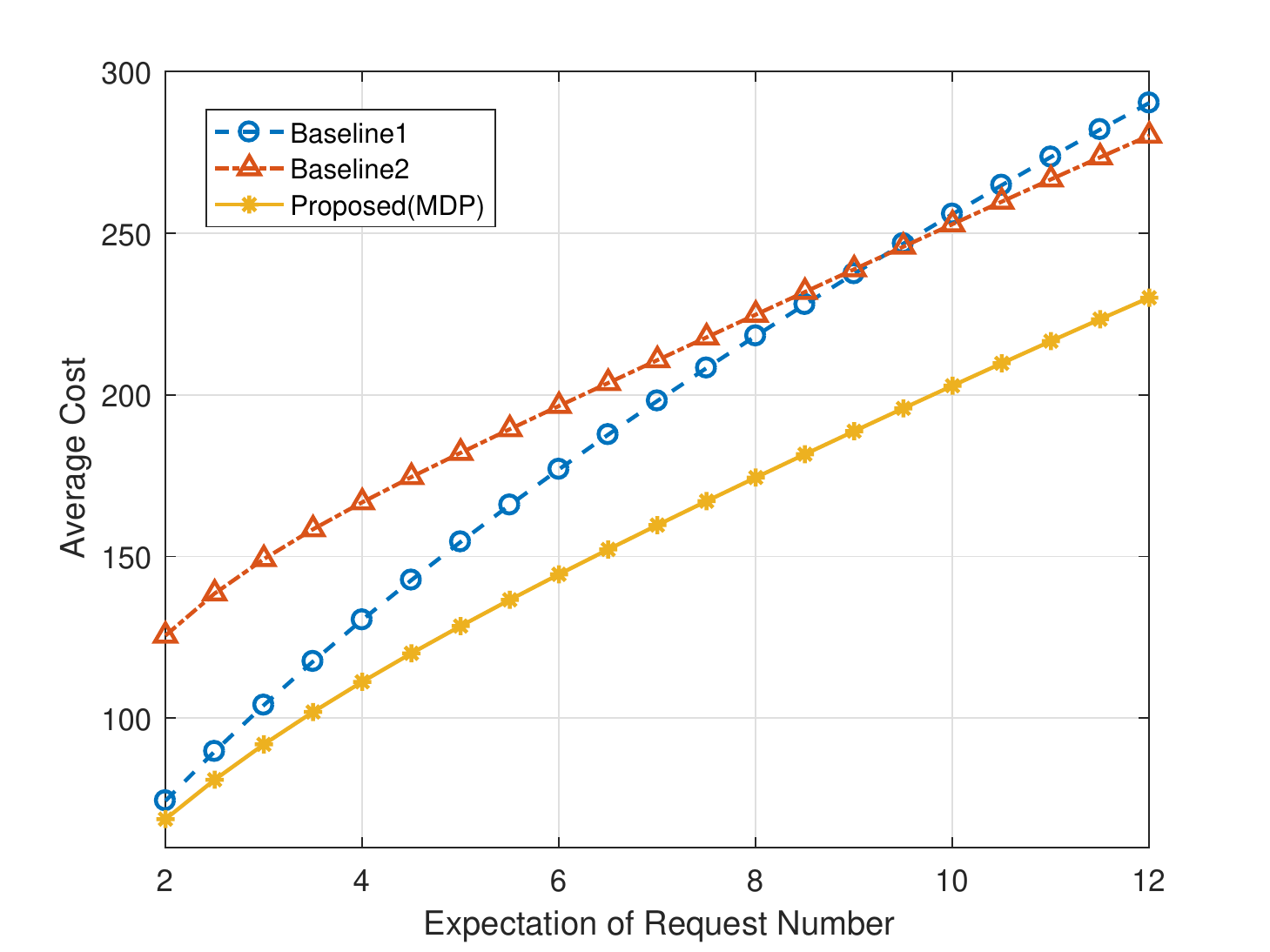}} 
	\hspace{0.2in} 
	\subfigure[$ N_C = 25 $]{ \label{fig:25caches} 
	\includegraphics[width=3in]{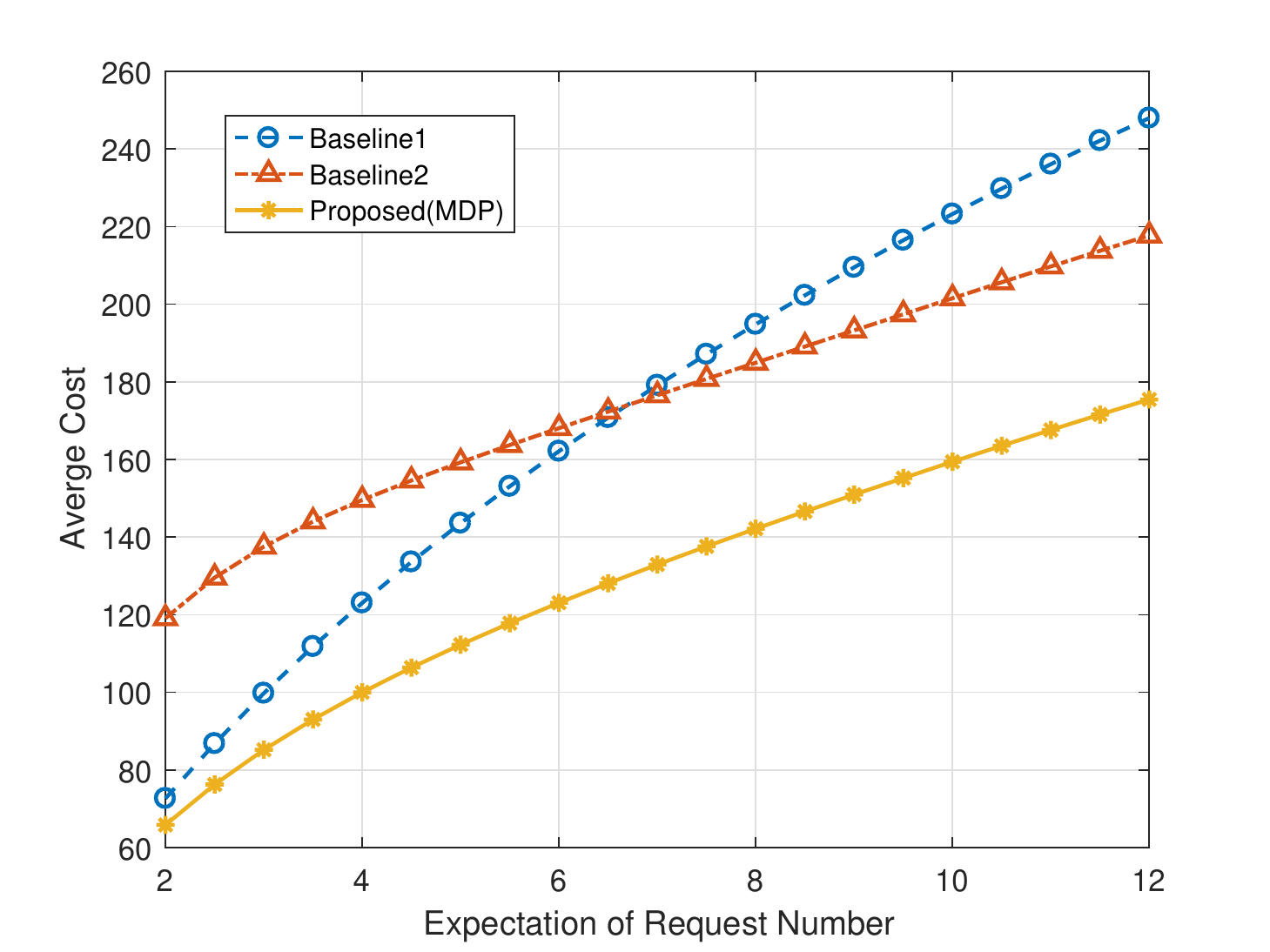}} 
	\caption{The average total cost versus the average times of request in file lifetime, where the number of cache nodes is $ 20 $ and $  25 $ respectively.} 
	\label{fig:caches} 
\end{figure}

In the simulation, the cell radius $ 500 $ meters, cache nodes are randomly deployed in the cell-edge region with a service radius of $ 90 $ meters. The number of BS antennas is 8. The path loss exponent is $3.5$. The file segment size $ R_f^I= 14 $Mb ($ \forall f $). The transmission bandwidth is $ 20 $MHz. The power constraint at the base station $P_B=46$ dBm. The performance of the proposed algorithm will be compared with the following two baselines. 
\begin{Baseline}
	The BS only makes sure that the segment delivery to the requesting users in each transmission. The cache nodes with better channel condition to the BS can decode the segments.
\end{Baseline}
\begin{Baseline}The BS ensures that all the cache nodes can decode the downlink file in the first transmission. Hence, all the cache nodes can help to forward the file since the second file request.
\end{Baseline}

The performance of the proposed low-complexity algorithm (Algorithm \ref{alg:AMDP}) is compared with the above two baselines in Fig.\ref{fig:caches}. In the simulation, the number of cache nodes is $ 20 $ and $25$ respectively, and the requesting users are uniformly distributed in the cell coverage, and distribution statistics are known to the BS. Hence, the analytical expressions derived in Section \ref{sub:app} can be used to calculate the approximated value functions. It can be observed that the proposed Algorithm \ref{alg:AMDP} is superior to the two baselines for any expected number of requests per file lifetime. Moreover, the Baseline 1 has better performance than Baseline 2 when the popularity of the file is high (larger expected number of file requests). The performance gain tends to be a constant when expectation of request number is large. This is because all the three schemes have the same performance as long as the files have been stored in all cache nodes. In other words, the gain of the proposed scheme lies in the phase of caching. 

\begin{figure} 
	\centering 
	\subfigure[$\widetilde{S} = \{\mathcal{B}^i_{f,s} = 0, \forall f,s, i\} $]{ \label{fig:bound_full}  
		\includegraphics[width=3in]{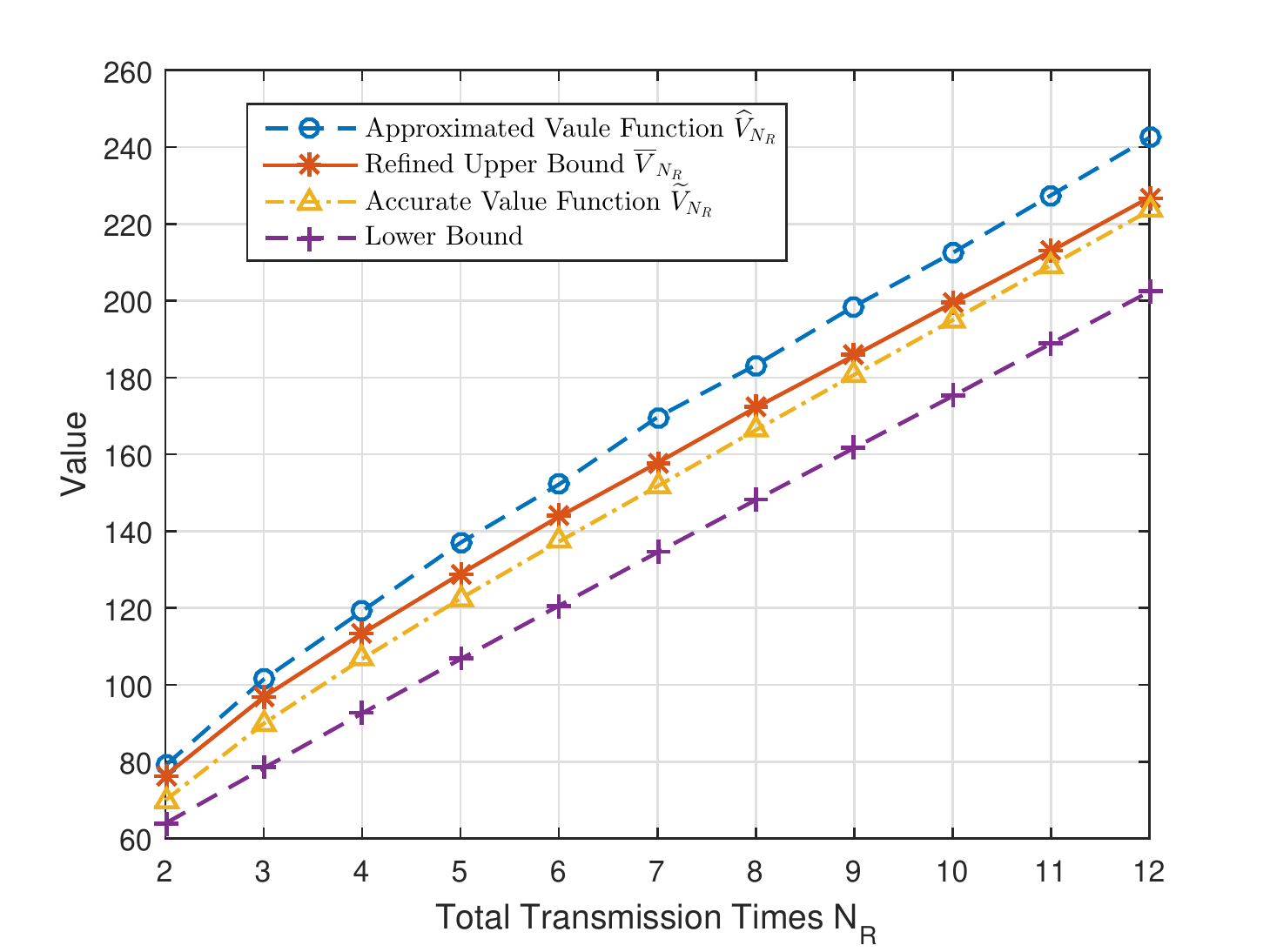}} 
	\hspace{0.2in} 
	\subfigure[$\widetilde{S} = \{\mathcal{B}^i_{f,s} = 0, \forall f,s, i=1,2,...,10\} \cup \{\mathcal{B}^i_{f,s} = 1, \forall f,s, i=11,12,...,20\}$]{ \label{fig:bound_half}
		\includegraphics[width=3in]{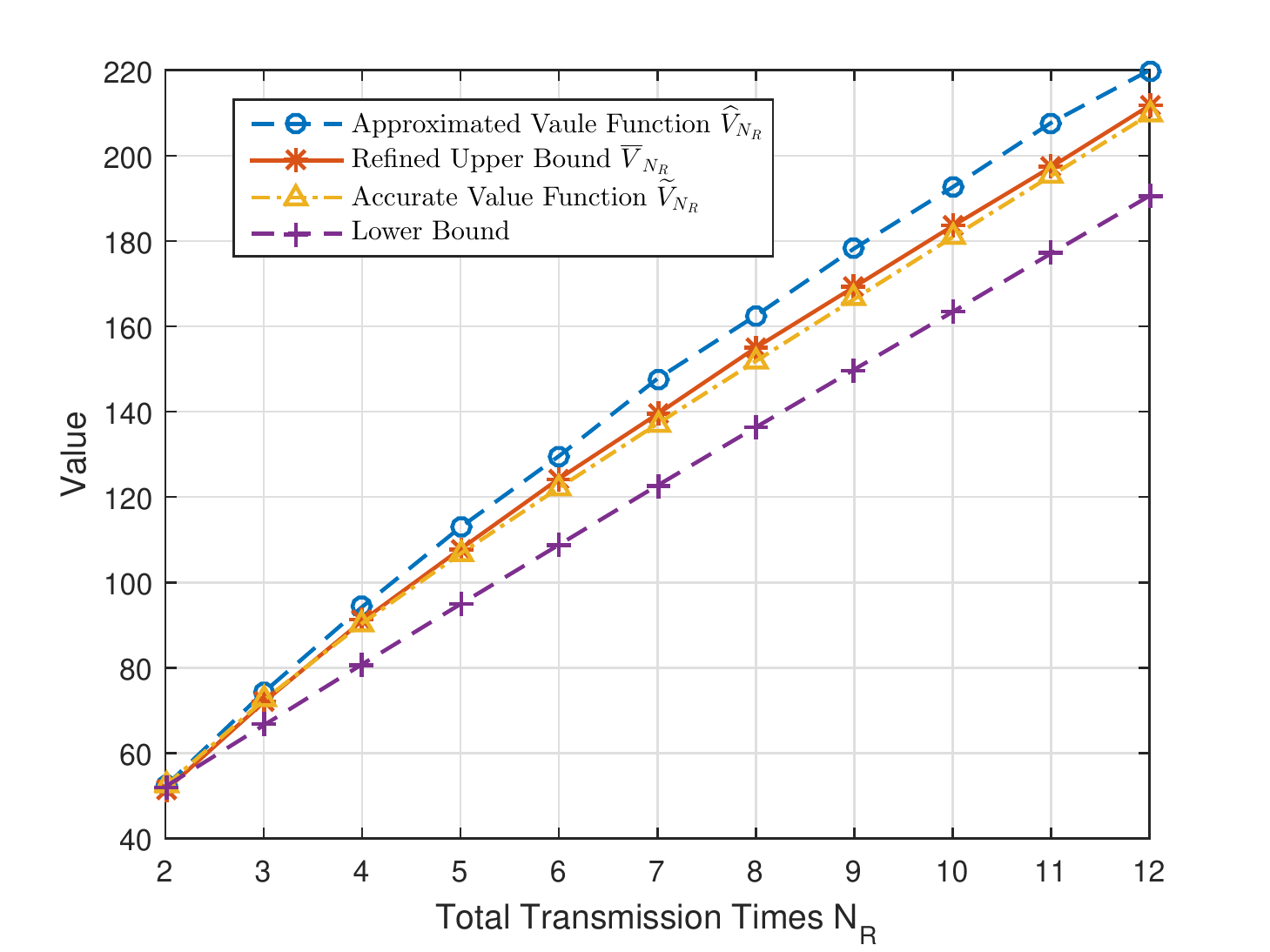}} 
	\caption{Illustration of value function and its bounds, where $ N_C = 20 $.} 
	\label{fig:bound} 
\end{figure}

The approximation error of the value function versus the indexes of file requests is illustrated in Fig.\ref{fig:bound}, where the true value function and the bounds derived in Lemma \ref{lem:bound} are plotted.
The cache nodes are empty in Fig. \ref{fig:bound_full}, while half of cache nodes have decoded the whole file in Fig. \ref{fig:bound_half}. It is shown that for both states, both upper and lower bounds are tight, and therefore the approximation error is small. In addition, the refined upper bound has even smaller gap to the true value function, which matches the conclusion in Corollary \ref{Corollary:refined_upper_bound}.

\begin{figure} 
	\centering 
	\subfigure[$ 3 $ hot zones]{ \label{fig:learning(3)}  
		\includegraphics[width=3in]{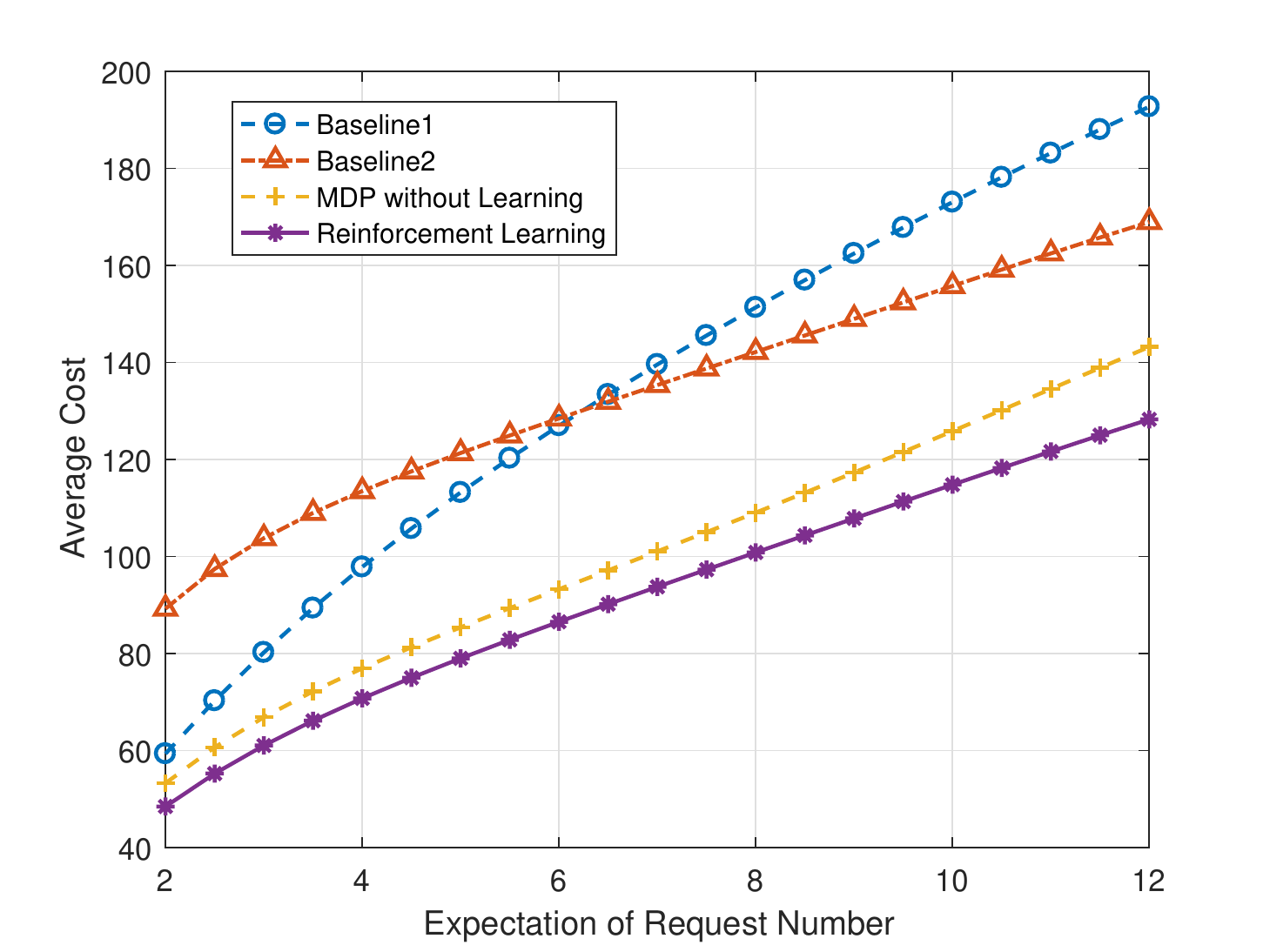}}
	\hspace{0.2in} 
	\subfigure[$ 4 $ hot zones]{ \label{fig:learning(4)}
		\includegraphics[width=3in]{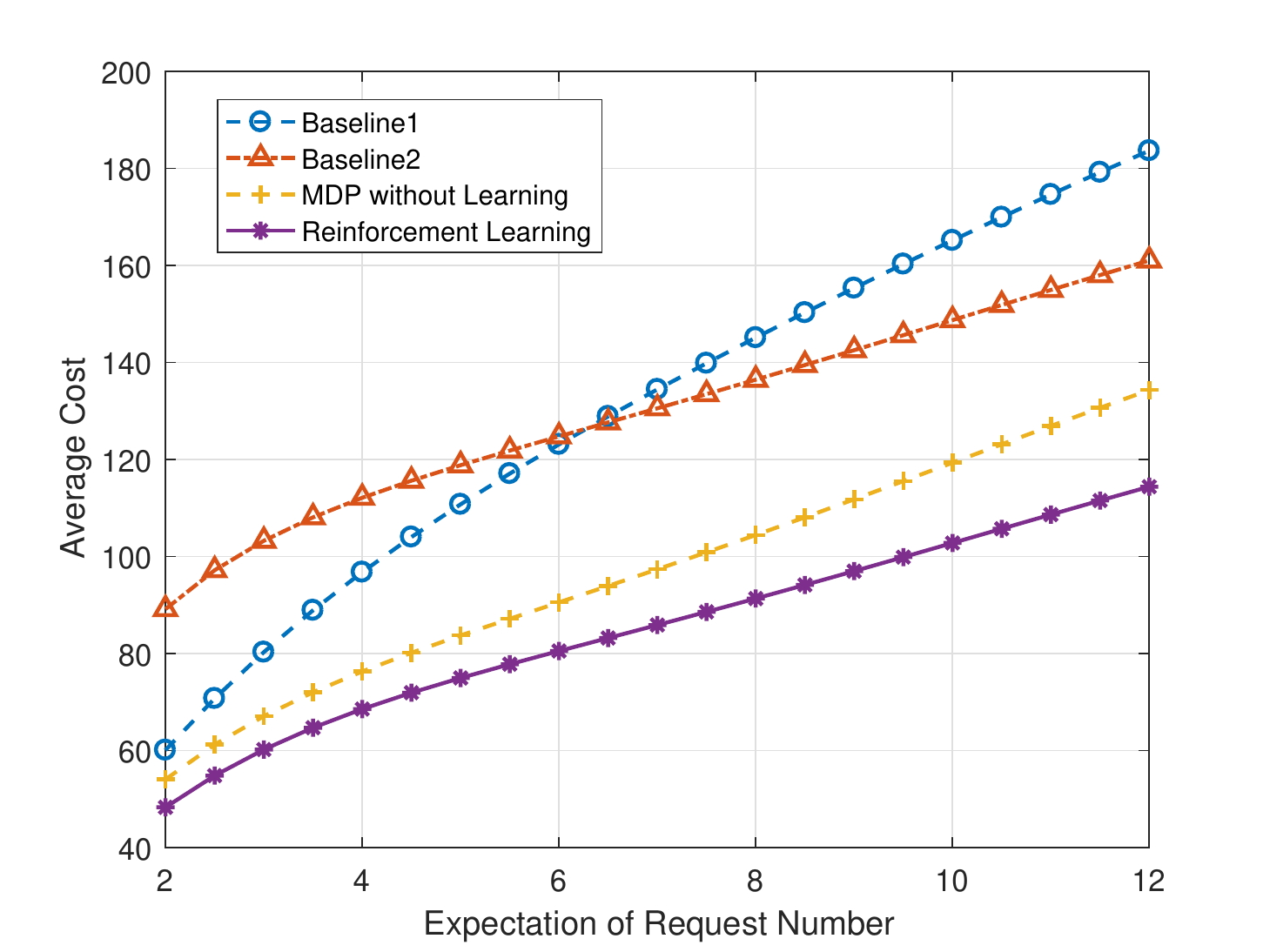}} 
	\caption{The average total cost versus the expectation of request times, where exits 3 hot zones in the cell.} 
	\label{fig:learn} 
\end{figure}

In Fig. \ref{fig:learn}, there are $ 3 $ and $ 4 $ hot zones in the cell coverage respectively, each with radius $ 90 $ m. The probability that the user appears in the one hot zone is $12.5\%$ (larger than the other regions). The locations and user distribution of the hot zones are unknown to the BS. The performance of two baselines, the proposed Algorithm \ref{alg:AMDP} assuming users are uniformly distributed, and the proposed Algorithm \ref{alg:AMDP} with learning-based evaluation of value functions (Algorithm \ref{alg:learning}) are compared. It can be observed that the proposed  learning algorithm has the best performance. Moreover, the performance gain of the learning-based algorithm is larger with more hot zones.

\begin{figure}
	\centering
		\subfigure[]{ \label{fig:proactive(a)}  
	\includegraphics[width=3in]{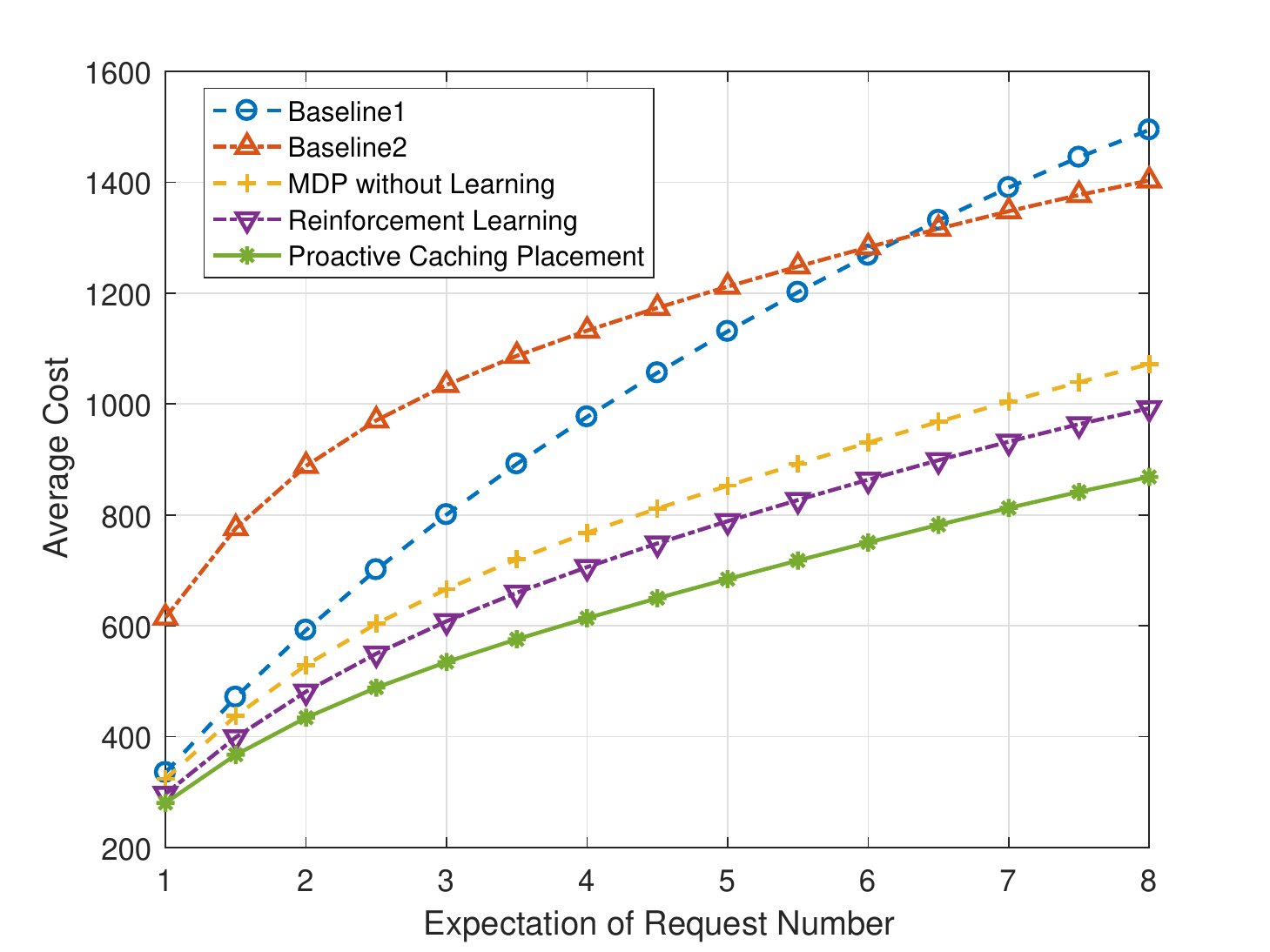}}
\hspace{0.2in} 
\subfigure[]{ \label{fig:proactive(b)}
	\includegraphics[width=3in]{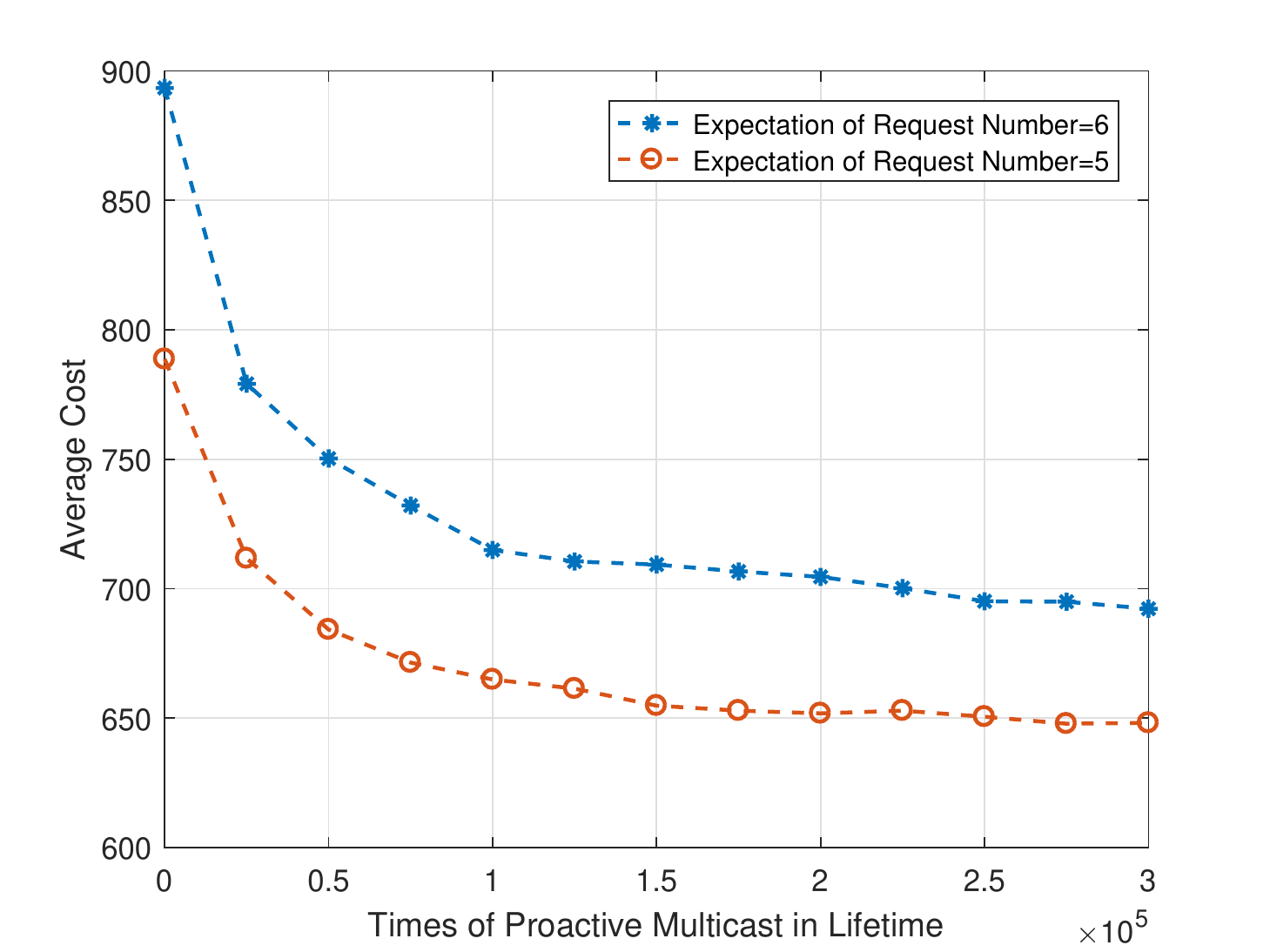}} 
	\caption{The performance of proactive multicast. In (a), the average system transmission cost versus the average number of requests per lifetime is illustrated. In (b), the average system transmission cost is illustrated for different times of proactive multicast.}
	\label{fig:proactive}
\end{figure}

Finally, the performance gain of the proactive multicast is demonstrated in Fig. \ref{fig:proactive} (a), where there are 10 files, 3 hot zones in the cell, 50000 times of proactive transmission opportunities in the file's lifetime. The performance of the proactive content placement algorithm is compared with the above two baselines, Algorithm \ref{alg:AMDP} assuming uniform user distribution, and Algorithm \ref{alg:AMDP} with learning-based evaluation of value functions. It can be observed that the proposed proactive content placement algorithm can further improve the offloading performance, especially when the popularity of files is high. Moreover, it is shown in Fig. \ref{fig:proactive} (b) that the average cost decreases with the increasing of proactive multicast frequency. However, the transmission cost reduction saturates when the proactive multicast frequency are large. 
	
\section{Conclusion}\label{sec:con}
We consider the scheduling of downlink file transmission with the assistance of cache nodes in this paper. The downlink resource minimization problem with reactive multicast is formulated as a dynamic programming problem with random number of stages. We first approximate it by a finite-horizon MDP with fixed stage numbers. In order to address the curse of dimensionality, we also introduce a low-complexity sub-optimal solution based on linear approximation of value functions. The approximated value function can be calculated analytically with the knowledge of distribution statistics of users. Since, the statistics of the distribution may be unknown to the BS, we continue to propose a learning-based online algorithm to evaluate the approximated value functions. Furthermore, we derive a bound on the gap between the approximated value functions and the real ones. Finally, we propose a proactive multicast algorithm, which can exploit the channel temporal diversity of shadowing effect. 

\section*{Appendix A: Proof Of Lemma \ref{Lemma:reduce_space}}
Let $\widetilde{V}_{N_R-n}(\widetilde{S}_{f,n})=\mathbb{E}_{\eta,\rho}[V_{N_R-n}(S_{f,n+1})] $, where the expectation is taken over the randomness of shadowing and requesting users' pathloss, we have
\begin{align*}
&\sum\limits_{S_{f,n+1}}V_{N_R-n}(S_{f,n+1})\Pr(S_{f,n+1}|S_{f,n},\Omega_{f,n})\\
=&\sum\limits_{\widetilde{S}_{f,n+1}}{\widetilde{V}_{N_R-n}(\widetilde{S}_{f,n+1})\Pr(\widetilde{S}_{f,n+1}|{S}_{f,n},\Omega_f)}
\end{align*}
Taking expectation with respect to the shadowing and pathloss on (\ref{eqn:bellman-fix}), we have
\begin{align*}
\widetilde{V}_{N_R-n+1}(\widetilde{S}_{f,n}) =& \mathbb{E}_{\eta,\rho} \bigg\{ \min_{\Omega(S_{f,n})} \sum_{s} g_{f,n,s}(S_{f,n},\Omega_{f,n})\\ &\!\!\!+\!\!\!\! \!\sum\limits_{S_{f,n+1}}\!\!\!\!{V_{N_R-n}(\!S_{f,n+1}\!)\!\Pr(\!S_{f,n+1}|S_{f,n},\Omega_{f,n}\!)} \! \bigg \} \\
=&\min_{\Omega(\widetilde{S}_{f,n})} \mathbb{E}_{\eta,\rho} \bigg\{\sum_{s} g_{f,n,s}(S_{f,n},\Omega_{f,n})\\
&\!\!\!+ \!\!\!\! \!\sum\limits_{\widetilde{S}_{f,n+1}}\!\!\!\!{\widetilde{V}_{N_R-n}(\!\widetilde{S}_{f,n+1}\!)\Pr(\widetilde{S}_{f,n+1}|{S}_{f,n},\Omega_f)}\! \bigg\} .
\end{align*}

\section*{Appendix B: Proof Of Lemma \ref{lem:cost_lower_bound}}
Due to page limitation, we only provide the sketch of the proof. 
\begin{align*}
& \mathbb{E}_{S_{f,n+1}}[W(S_{f,n+1},T_{f,n}^f)|S_{f,n}] \nonumber\\ =&\mathbb{E}_{S_{f,n+1}} \bigg \{\!\min\limits_{\{\Omega_{f,k}|\forall k=n+1,...\}} \sum_N \mathbb{E}_{\eta,\rho,{\mathcal{T}}} \bigg[ \frac{(\lambda_f T_{f,n}^f)^N}{N!} e^{-\lambda_f T_{f,n}^f}\\
 &\times \sum_{n=1}^{N} \sum_{s=1}^{N_f} g_{f,n,s}(\Omega_{f,n}) \bigg|S_{f,n+1}\bigg] \bigg \}\\
\overset{(a)}{\geq}& \mathbb{E}_{S_{f,n+1}} \bigg \{ \sum_N \frac{(\lambda_f T_{f,n}^f)^N}{N!} e^{-\lambda_f T_{f,n}^f}\\
&\times\underbrace{\min\limits_{\{\Pi_{f,k}^N|\forall k\}}\mathbb{E}_{\eta,\rho,{ \mathcal{T}}} \bigg[  \sum_{n=1}^{N} \sum_{s=1}^{N_f} g_{f,n,s}(\Omega_{f,n}) \bigg|S_{f,n+1}\bigg] }_{\widetilde{V}_{N}(\widetilde{S}_{f,n+1})}\bigg \} \\
=& \!\!\!\!\sum\limits_{N,\widetilde{S}_{f,n+1}} \!\! \!\!\!\!\!\frac{(\lambda_f T_{f,n}^f)^N}{N!} e^{-\lambda_f T_{f,n}^f}{\widetilde{V}_{N}(\!\widetilde{S}_{f,n+1}\!)\Pr(\!\widetilde{S}_{f,n+1}|{S}_{f,n},\Omega_{f,n}\!)}, 
\end{align*}
where $\{\Pi_{f,k}^N|\forall k\}$ is the optimal policy when the remaining stage number is $ N $. The inequality (a) is because that $\{\Pi_{f,k}^N\}$ is optimized for each specific remaining stage number.

\section*{Appendix C: Proof Of Lemma \ref{Lemma:Segment} }
First of all, we have the following high SINR  approximation on the throughput $ R_{f,n,s} $. $R_{f,n,s} \approx N_{f,n,s} \mathbb{E}_{\mathbf{h}_{f,n,s}} \left[\alpha \log_2 \left(  \frac{||\mathbf{h}_{f,n,s}||^2 P_{f,n,s}}{N_T \sigma^2_z} \right) \right]=N_{f,n,s}\alpha[\theta+\log_2(P_{f,n,s})]$.
With $R_{f,n,s}= R_f^I$, we  have $
N_{f,n,s}=\frac{R_f^I}{\alpha[\theta+\log_2(P_{f,n,s})]}$.
Hence the original optimization becomes
$
\min \limits_{P_{f,n,s} }  \frac{R_f^I(P_{f,n,s}+w)}{\alpha[\theta+\log_2(P_{f,n,s})]}. \nonumber
$
The optimal transmission power $ P_{f,n,s}^* $ can be obtained by taking first-order derivative.

\section*{Appendix D: Proof Of Lemma \ref{Lemma:learning} }
We only prove the convergence of $ \widetilde{V}^t_{N_R-n+1}(\widetilde{S}_{f}^*) $, and the convergence of $ \widetilde{V}^t_{N_R-n+1}(\widetilde{S}_{f}^{i,s}) $ can be applied similarly.
Let $\varepsilon_t = \widetilde{V}_{N_R-n+1}(\widetilde{S}_{f}^*) - \frac{(N_R-n+1)R_f^I}{R_g^I}I(\mathbf{l}_{g,m} \notin \mathcal{C})
\sum_s ( P_{g,m,s}^*N_{g,m,s}^* + w N_{g,m,s}^*)$ be the estimate error in $t$-th iteration. It is clear that the estimation errors are i.i.d. with respect to $ t $,  $\mathbb{E} [\varepsilon_t] =0$ and $Var [\varepsilon_t ] < +\infty$. Note that  $\widetilde{V}^t_{N_R-n+1}(\widetilde{S}_{f}^*)$ can be written as
\begin{eqnarray}
\widetilde{V}^t_{N_R-n+1}(\widetilde{S}_{f}^*)&=&\sum_{i=0}^{t} \frac{\widetilde{V}_{N_R-n+1}(\widetilde{S}_{f}^*)-\varepsilon_i}{t+1} \nonumber \\
&=&\widetilde{V}_{N_R-n+1}(\widetilde{S}_{f}^*)-\sum_{i=0}^{t}\frac{\varepsilon_i}{t+1} ,\nonumber
\end{eqnarray}
where the total estimate error is $\sum_{i=0}^{t}\frac{\varepsilon_i}{t+1}$. The mean and variance of total estimate error are $
\mathbb{E}\bigg\{     \sum_{i=0}^{t}\frac{\varepsilon_i}{t+1}                 \bigg\}=0,   Var\bigg\{     \sum_{i=0}^{t}\frac{\varepsilon_i}{t+1}\bigg\} = \frac{Var[\varepsilon_i ]}{t+1}. $
When $t\rightarrow +\infty $,  the variance of estimation error tends to zero, and $ \widetilde{V}^t_{N_R-n+1}(\widetilde{S}_{f}^*) $ converges to $ \widetilde{V}_{N_R-n+1}(\widetilde{S}_{f}^*) $.

\section*{Appendix E: Proof Of Lemma \ref{lem:bound} }
\subsubsection{Proof of Upper Bound}
The approach of mathematical induction will be used in the proof. Without loss of generality, we shall assume that  the upper bound holds when the first $ l $-th cache nodes have not decoded the $ (f,s) $-th segment, and prove that the upper bound also holds when the first $ (l+1) $-th cache nodes have not decoded the $ (f,s) $-th segment.
Define the system state $  \widetilde{T_f}^{c,s} = [\mathcal{B}^i_{f,j}=1,\forall j\neq s, \forall i] \cup[\mathcal{B}^i_{f,s}=0,\forall i=1,2,\cdots,c]\cup[\mathcal{B}^i_{f,s}=1,\forall i > c]$.

\begin{itemize}
	\item Step 1: When $c=1$, the upper bound holds as follows
	\begin{align*}
	\!\widetilde{V}_{N_R-n+1}(\widetilde{T_f}^{c,s}) =& \widetilde{V}_{N_R-n+1}(\widetilde{S}_{f}^*) \\ &\!\!\!+\!\!\bigg( \! \widetilde{V}_{N_R-n+1}(\widetilde{S}_{f}^{c,s})-\widetilde{V}_{N_R-n+1}(\widetilde{S}_{f}^*)\!\bigg) 
	\end{align*}
	
	\item Step 2: Suppose the following bound holds for $ c=l $
		\begin{align*}\label{eqn:Assumption1}
		\!\widetilde{V}_{N_R-n+1}(\!\widetilde{T_f}^{l,s}) \!\leq & \widetilde{V}_{N_R-n+1}(\widetilde{S}_{f}^*)\\&\!\!\!\!+ \!\!\!\!\!\!\!\!\!\sum_{j=1,2,\cdots,l} \!\!\!\bigg(\!\! \widetilde{V}_{N_R-n+1}(\!\widetilde{S}_{f}^{c,j})\!-\!\widetilde{V}_{N_R-n+1}(\!\widetilde{S}_{f}^*\!)\!\!\bigg) 
		\end{align*}
	
	\item Step 3: When $c=l+1$, we can apply the following sub-optimal control policy: (1) if the requesting users appear in the coverage of $ \mathcal{C}_1 \cup \mathcal{C}_2 \cup ... \cup \mathcal{C}_l $, the optimal scheduling policy for system state $ \widetilde{T_f}^{l,s} $ is applied; (2) if the requesting users appear in the coverage of $ \mathcal{C}_{l+1} $, the optimal scheduling policy for system state $ \widetilde{S}_f^{l+1, s} $ is applied; (3) if the requesting users appear outside the coverage of any cache nodes, choose the one from the above two policies with larger transmission resource consumption. Let $\breve{V}_{N_R-n+1}(\widetilde{T_f}^{l+1,s})  $ be the average cost of the above sub-optimal scheduling policy, we have
	\begin{align*}
	\!\widetilde{V}_{N_R-n+1}(\!\widetilde{T_f}^{l+1,s})\!\! \leq& \breve{V}_{N_R-n+1}(\widetilde{T_f}^{l+1,s}) \\\leq& \widetilde{V}_{N_R-n+1}(\!\widetilde{T_f}^{l,s})\\ &\!\!\!\!+\!\! \bigg(\!\! \widetilde{V}_{N_R-n+1}(\widetilde{S}_{f}^{c,l+1})\!-\!\widetilde{V}_{N_R-n+1}(\!\widetilde{S}_{f}^*\!)\!\!\bigg).
	\end{align*}

\end{itemize}

Although the above proof is for the $ (f,s) $-th file segment, it can be trivially extended to arbitrary file segments. Thus the upper bound is proved.

\subsubsection{Proof of Lower Bound}
Let $\Omega_{f,n}^*$ be the optimal scheduling policy and $\widetilde{S}_{f,n}$ be arbitrary cache state for the $f$-th file in the $n$-th stage, we have 
\begin{align*}
&\widetilde{V}_{N_R-n+1}(\widetilde{S}_{f,n})-\widetilde{V}_{N_R-n+1}(\widetilde{S}_{f}^*) \\
= &\mathbb{E}_{\eta,\rho} \bigg\{ \sum_{s} g_{f,n,s}(\widetilde{S}_{f,n},\Omega_{f,n}^*)\bigg\}\\ 
&+\mathbb{E}_{\eta,\rho} \bigg\{\sum\limits_{\widetilde{S}_{f,n+1}}{\widetilde{V}_{N_R-n}(\widetilde{S}_{f,n+1})\Pr(\widetilde{S}_{f,n+1}|{S}_{f,n},\Omega_{f,n}^*)}\bigg\}  \\
& -  \mathbb{E}_{\eta,\rho} \bigg\{ \sum_{s} g_{f,n,s}(\widetilde{S}_{f}^*,\Omega_{f,n}^*)\bigg\}-\widetilde{V}_{N_R-n}(\widetilde{S}_{f}^*) \nonumber.
\end{align*}
As $$
\mathbb{E}_{\eta,\rho} \bigg\{\!\! \sum\limits_{s} g_{f,n,s}(\!\widetilde{S}_{f,n},\Omega_{f,n}^*\!)\!\bigg\} \!\geq\! \mathbb{E}_{\eta,\rho} \bigg\{\!\! \sum\limits_{s} g_{f,N_R,s}(\!\widetilde{S}_{f,n},\Omega_{f,N_R}^*\!)\!\bigg\}$$  and $$\mathbb{E}_{\eta,\rho} \bigg\{\!\! \sum\limits_{s} g_{f,n,s}(\widetilde{S}_{f}^*,\Omega_{f,n}^*)\!\bigg\} \!\!=\! \mathbb{E}_{\eta,\rho} \bigg\{\! \sum\limits_{s} g_{f,N_R,s}(\widetilde{S}_{f}^*,\Omega_{f,N_R}^*)\!\bigg\},$$
We have
\begin{align*}
&\mathbb{E}_{\eta,\rho} \bigg\{ \sum_{s} g_{f,n,s}(\widetilde{S}_{f,n},\Omega_{f,n}^*)\bigg\} -\mathbb{E}_{\eta,\rho} \bigg\{ \sum_{s} g_{f,n,s}(\widetilde{S}_{f}^*,\Omega_{f}^*)\bigg\} \\
\geq&\mathbb{E}_{\eta,\rho}\! \bigg\{\! \!\!\sum_{s}\! g_{f,N_R,s}(\!\widetilde{S}_{f,n},\!\Omega_{f,N_R}^*\!)\!\!\bigg\} \!-\! \mathbb{E}_{\eta,\rho}\! \bigg\{ \!\!\!\sum_{s}\! g_{f,N_R,s}(\!\widetilde{S}_{f}^*,\!\Omega_{f,N_R}^*\!)\!\!\bigg\}\\
=&\!\!\!\!\sum_{{\{(i,s)|\forall \mathbf{B}^i_{f,s}(\widetilde{S}_{f,n}) = 0\}}}\!\!\!\! \bigg( \widetilde{V}_{1}(\widetilde{S}_{f}^{i,s})-\widetilde{V}_{1}(\widetilde{S}_{f}^*)\bigg).\nonumber
\end{align*}
We also have 
\begin{align*}
\!\mathbb{E}_{\eta,\rho} \bigg\{\sum\limits_{\widetilde{S}_{f,n+1}}{\widetilde{V}_{N_R-n}(\widetilde{S}_{f,n+1})\Pr(\widetilde{S}_{f,n+1}|{S}_{f,n},\Omega_{f,n}^*)}\bigg\}&\\ -\widetilde{V}_{N_R-n}(\widetilde{S}_{f}^*)
&\geq 0.
\end{align*}
  As a result, the lower bound is straightforward.

\section*{Appendix F: Proof Of Corollary \ref{Corollary:refined_upper_bound} }
Due to page limitation, we only provide the sketch of the proof. First, $  \overline{V}_{N_R-n+1} (\widetilde{S}_{f,n}) \geq \widetilde{V}_{N_R-n+1} (\widetilde{S}_{f,n})$ can be deduced from the following two factors:
\begin{itemize}
	\item It has been proved in Lemma \ref{lem:bound} that $\widehat{V}_{N_R-n}(\widetilde{S}_{f,n+1})\geq\widetilde{V}_{N_R-n}(\widetilde{S}_{f,n+1}),\forall \widetilde{S}_{f,n+1}$.
	
	\item $  \overline{V}_{N_R-n+1} (\widetilde{S}_{f,n}) $ is the minimization of the $(f,n)$-th file transmission cost and the future cost $ \widehat{V}_{N_R-n}(\widetilde{S}_{f,n+1}) $; whereas $  \widetilde{V}_{N_R-n+1} (\widetilde{S}_{f,n}) $ is the minimization of the $(f,n)$-th file transmission cost and the future cost $ \widetilde{V}_{N_R-n}(\widetilde{S}_{f,n+1}) $. 
\end{itemize}

\begin{figure*}
	\begin{equation}\label{eqn:refin_proof}
	\overline{V}_{N_R-n+1} (\widetilde{S}_{f,n}) \overset{(a)} { \leq }\underbrace{\mathbb{E}_{\eta,\rho} \bigg\{\sum_{s}\! g_{f,n,s}(\widehat{\Omega}_{f,n}) + \sum\limits_{\widetilde{S}_{f,n}}{\widehat{V}_{N_R-n}(\widetilde{S}_{f,n+1})\!\Pr(\widetilde{S}_{f,n+1}|{S}_{f,n},\widehat{\Omega}_{f,n})}\bigg\} }_{\mbox{Denoted as }\check{V}_{N_R-n+1}(\widetilde{S}_{f,n})}\! \overset{(b)} {\leq}  \widehat{V}_{N_R-n+1}(\widetilde{S}_{f,n}) 
	\end{equation}
	\hrulefill
\end{figure*}

In order to prove $ \overline{V}_{N_R-n+1} (\widetilde{S}_{f,n}) \leq \widehat{V}_{N_R-n+1} (\widetilde{S}_{f,n})$, we first define the scheduling policy $\widehat{\Omega}_{f,n}$ as follows.
\begin{itemize}
	\item When the first requesting user falls in the coverage area of the $ c $-th cache node ($\mathbf{l}_{f,n} \in \mathcal{C}_{c}$, $ \forall c $), the scheduling policy $\widehat{\Omega}_{f,n}=\{(\widehat{P}_{f,n,s}^c,\widehat{N}_{f,n,s}^c)|\forall s\}$ minimizes the transmission cost for the $ s $-th segment ($ \forall s $) as if all other cache nodes have already decoded this segment. Hence, 
	\begin{align*}
	\!(\widehat{P}_{f,n,s}^c,\widehat{N}_{f,n,s}^c)=&\! \arg \! \min\limits_{P_{f,n,s} \atop N_{f,n,s}}\!\!\! \bigg\{ \! g_{f,n,s}\bigg(\!\widetilde{S}_{f,n},P_{f,n,s},N_{f,n,s}\!\bigg) \\
	 &\!+\!\!\! \sum \! \Pr(\widetilde{S}_{f,n+1}|\widetilde{S}_{f,n})\widehat{V}_{N_R-n}(\widetilde{S}_{f,n+1}) \! \bigg\}\\
	s.t.& \quad R_{f,n,s} = R_f^I \mbox{ and } \widehat{P}_{f,n,s}^c\leq P_B, \nonumber
	\end{align*}
	\item When the first requesting user falls outside the coverage area of any cache node ($\mathbf{l}_{f,1} \in \mathcal{C}_{0}$), the scheduling policy $\widehat{\Omega}_{f,n}=\{(P_{f,n,s}^0,N_{f,n,s}^0)|\forall s\}$, where 
	\begin{align*}
	&P_{f,n,s}^0= P_{f,n,s}^*+\sum_{ \forall c } (P_{f,n,s}^c-P_{f,n,s}^*),\\ &N_{f,n,s}^0=N_{f,n,s}^*+\sum_{ \forall c  } (N_{f,n,s}^c-N_{f,n,s}^*), 
	\end{align*}
	$ P_{f,n,s}^* $ and $ N_{f,n,s}^* $ are the optimal scheduling by assuming that all the cache nodes have decoded the $ s $-th segment (the expressions of them are provided in Lemma \ref{Lemma:Segment}). $ P_{f,n,s}^c $ and $ N_{f,n,s}^c $ are the optimal scheduling by assuming that all the cache nodes except the $ c $-th one have decoded the $ s $-th segment (the expressions of them are provided in Lemma \ref{Lemma:default_one_Segment}).
	
\end{itemize}

Then, we get the inequalities \eqref{eqn:refin_proof}, where $ \widetilde{S}_{f,n+1} $ is the next cache state given the current cache state $ {S}_{f,n} $ and scheduling policy $ \widehat{\Omega}_{f,n} $. The inequality (a) is because that $\overline{V}_{N_R-n+1} (\widetilde{S}_{f,n})$ uses the optimal scheduling policy and $\check{V}_{N_R-n+1}(\widetilde{S}_{f,n})$ uses heuristic scheduling policy. Both  $\check{V}_{N_R-n+1}(\widetilde{S}_{f,n})$ and $\widehat{V}_{N_R-n+1}(\widetilde{S}_{f,n})$ spent the same cost on the first file transmission. However, the evaluation of future cost in  $\widehat{V}_{N_R-n+1}(\widetilde{S}_{f,n})$ is more conservative (larger) than that of $\check{V}_{N_R-n+1}(\widetilde{S}_{f,n})$. The inequality (b) can be obtained.

\bibliographystyle{IEEEtran}
\bibliography{TCOM_Cache}

\end{document}